\documentclass[letterpaper, leqno, 11pt]{article}
\usepackage[utf8]{inputenc}

\usepackage[letterpaper, margin=1in]{geometry}

\usepackage{hyperref}

\usepackage{amssymb,amsmath,amsfonts,amsthm}
\usepackage{graphicx,hyperref,xcolor,xspace}
\usepackage{mathtools}
\usepackage{todonotes}
\usepackage{paralist}
\usepackage{csquotes}
\usepackage{thm-restate}
\usepackage{algorithm} 

\usepackage[noend]{algpseudocode}
\algrenewcommand\textproc{}

\newtheorem{theorem}{Theorem}
\newtheorem{lemma}[theorem]{Lemma}
\newtheorem{observation}[theorem]{Observation}
\newtheorem{Definition}[theorem]{Definition}

\newtheorem{corollary}[theorem]{Corollary}

\newcommand{\Oh}{\ensuremath{\mathcal{O}}}
\newcommand{\poly}{\ensuremath{\text{poly}}}
\newcommand{\RR}{\mathbb{R}}
\DeclareMathOperator{\df}{d_{F}}
\def \eps{\varepsilon}
\DeclareMathOperator*{\argmin}{arg\,min}
\DeclareMathOperator*{\argmax}{arg\,max}
\newcommand{\seqtocurve}[1]{\left\langle#1\right\rangle}
\def\radius{\delta}
\def\HHH{{\mathcal H}}
\def\GGG{{\mathcal G}}
\def\PPP{{\mathcal P}}
\def\CCC{{\mathcal C}}

\def\SSS{{\mathcal S}}
\def\III{{\mathcal I}}
\def\ZZ{{\mathbb Z}}
\def\NN{{\mathbb N}}
\newcommand{\straightening}{straightening\xspace}
\newcommand{\straightenings}{straightenings\xspace}
\newcommand{\Straightenings}{Straightenings\xspace}
\newcommand{\deltamonotone}{$\delta$-monotone\xspace}
\newcommand{\ossov}{\textsc{OneSidedSparseOV}\xspace}

\begin{document}
\title{Tight Bounds for Approximate Near Neighbor Searching for Time Series under the Fr\'echet Distance}

\author{Karl Bringmann\thanks{Saarland University and Max Planck Institute for Informatics,  Saarland Informatics Campus, Saarbrücken, Germany. This work is part of the project TIPEA that has received funding from the European Research Council (ERC) under the European Unions Horizon 2020 research and innovation programme (grant agreement No. 850979).} \and Anne Driemel\thanks{Hausdorff Center for Mathematics, University of Bonn, Germany} \and André Nusser\thanks{Saarbrücken Graduate School of Computer Science and Max Planck Institute for Informatics,  Saarland Informatics Campus, Saarbrücken, Germany} \and Ioannis Psarros\thanks{Athena Research Center, Athens, Greece. Part of this work was done while the author was a postdoctoral researcher at the Hausdorff Center for Mathematics and the  University of Bonn, Germany. The author was (partially) supported by the European Union's Horizon 2020 Research and Innovation programme, under the grant agreement No. 957345: “MORE”. }}

\date{}

\maketitle

\begin{abstract}
\noindent
We study the $c$-approximate near neighbor problem under the continuous Fr\'echet distance: Given a set of $n$ polygonal curves with $m$ vertices, a radius $\delta > 0$, and a parameter $k \leq m$, we want to preprocess the curves into a data structure that, given a query curve $q$ with $k$ vertices, either returns an input curve with Fr\'echet distance at most $c\cdot \delta$ to $q$, or returns that there exists no input curve with Fr\'echet distance at most $\delta$ to $q$. We focus on the case where the input and the queries are one-dimensional polygonal curves---also called time series---and we give a comprehensive analysis for this case. We obtain new upper bounds that provide different tradeoffs between approximation factor, preprocessing time, and query time.

Our data structures improve upon the state of the art in several ways. We show that for any $0 < \varepsilon \leq 1$ an approximation factor of $(1+\varepsilon)$ can be achieved within the same asymptotic time bounds as the previously best result for $(2+\varepsilon)$. Moreover, we show that an approximation factor of $(2+\varepsilon)$ can be obtained by using preprocessing time and space $O(nm)$, which is linear in the input size, and query time in $O(\frac{1}{\varepsilon})^{k+2}$, where the previously best result used preprocessing time in $n \cdot O(\frac{m}{\varepsilon k})^k$ and query time in $O(1)^k$. We complement our upper bounds with matching conditional lower bounds based on the Orthogonal Vectors Hypothesis. 
Interestingly, some of our lower bounds already hold for any super-constant value of $k$. This is achieved by proving hardness of a one-sided sparse version of the Orthogonal Vectors problem as an intermediate problem, which we believe to be of independent interest.
\end{abstract}

\thispagestyle{empty}
\clearpage
\setcounter{page}{1}
\section{Introduction}

The Fr\'echet distance and its variants provide a versatile class of distance measures for parametrized curves as they occur in application areas such as trajectories of moving objects (e.g., vehicles, animals, or robots), outlines of shapes, signatures, gestures, and other types of time series from sensor data~\cite{acmsurvey20, su2020survey}. This distance measure is very similar to the Hausdorff distance, which is defined for sets, except that it takes the ordering of points along the curve into account. At the same time, by assuming the equivalence class of all reparametrizations of a curve, it is robust to local irregularities in the parametrization of the curves (e.g., errors due to local time delays or irregular measurements). An intuitive definition of the distance measure is given as follows. Imagine traversing the two curves at independent and varying speeds from the beginning to the end, and consider the maximum (Euclidean) distance that the two positions can maintain throughout the traversal without backtracking along the curves. Minimizing over all possible traversals yields the Fr\'echet distance of the two curves. 

Due to the popularity of the distance measure for trajectory analysis and data analysis applications, many heuristics and algorithm engineering solutions have been proposed to speed up the distance computation and similarity retrieval~\cite{BergGM17, BaldusB2017,DutschV17,BuchinDDM17,BringmannKN19, GHPS20}.
A fundamental task in this area is near neighbor searching: Preprocess $n$ curves into a data structure, such that we can query this data structure with a curve and retrieve an input curve that has small distance to the query curve. 
This problem has been studied intensively~\cite{BergCG13, BergMO17, I02, AD18, EP20, FiltserFK20, DS17, M20, DP21} and, for the discrete version of the Fr\'echet distance, these efforts lead to a simple and likely optimal data structure~\cite{FiltserFK20}. However, for the more classic continuous version of the Fr\'echet distance, the computational complexity of near neighbor searching is still largely open, and seems very challenging to resolve. 

Therefore, in this paper we focus on the special case of one-dimensional curves, which we also refer to as time series. We aim to resolve approximate near neighbor searching for this special case of the continuous Fr\'echet distance. We obtain strong lower bounds based on the Orthogonal Vectors Hypothesis in the regime of small approximation factors. More specifically, we differentiate a range of lower bounds for different approximation factors and preprocessing/query time. We show that our bounds are tight by devising data structures that asymptotically match the lower bounds in all cases considered. The new data structures improve upon the state of the art in several ways. For the same preprocessing and query time, we can improve the approximation factor from $(2+\varepsilon)$ to $(1+\varepsilon)$. For the same approximation factor $(2+\varepsilon)$, we  get a better time complexity---in some cases we can even achieve linear preprocessing time and space.

\subsection{Problem Definition}

Let us first formally define the distance measure considered in this work.
\begin{Definition}[Fr\'{e}chet distance]
Given two curves $P,Q:~ [0,1] \mapsto \RR^d $, their Fr\'{e}chet distance is
\[
	\df(P,Q) \coloneqq \min\limits_{f,g \in \mathcal{T}} \max_{t \in [0,1]} \|P(f(t)) - Q(g(t))\|_2,
\]
where $\mathcal{T}$ is the set of all monotone and surjective functions from $[0,1]$ to $[0,1]$. For functions $f$ and $g$ that realize the minimum above, we define $\phi: [0,1] \to [0,1]^2$ with $\phi(t) = (f(t), g(t)), t \in [0,1]$ and we refer to $\phi$ as a realizing traversal of the two curves. 
\end{Definition}

The central problem of this work is then defined as follows.
 \begin{Definition}[$c$-Approximate Near Neighbors problem ($c$-ANN)]\label{Dgenann}
The input consists of a set $\PPP$ of $n$ curves in $\RR^d$, each of complexity $m$, and a number $2 \le k \le m$. Given a distance threshold $\radius>0$ and an approximation factor $c>1$, preprocess $\PPP$ into a data structure such that for any query curve $Q$ of complexity $k$, the data structure reports as follows:
\begin{compactitem}
    \item if $\exists P\in \PPP$ such that $\df(P,Q)\leq \radius$, then it returns $P'\in \PPP$ such that $\df(P',Q)\leq c\radius$,
    \item if $\forall P \in \PPP$, $\df(P,Q)\geq c\radius$ then it returns “no”,
    \item otherwise, it either returns a curve $P\in \PPP$ such that $\df(P,Q)\leq c\radius$, or  “no”. 
\end{compactitem}
\end{Definition}

The assumption that all input curves have the same number of vertices $m$ and that the queries have $k$ vertices is mostly to simplify presentation; all our data structures are easily generalized to allow input curves of complexity \emph{at most} $m$ and query curves of complexity \emph{at most} $k$. Note, however, that we assume the input has size in $\Omega(nm)$ and that $2\leq k\leq m$. The case $k=1$ is a boundary case that is easier to solve; we ignore it throughout this paper.  

\subsection{State of the Art}


We start by reviewing the state of the art for the \emph{discrete} variant of the Fréchet distance. In the discrete Fr\'echet distance, the continuous traversal $\phi$ is replaced by a discrete traversal of the two point sequences, we refer to~\cite{FiltserFK20} for the exact definition. The currently best known data structure for $(1+\eps)$-ANN under the discrete Fr\'echet distance is by
Filtser et al.~\cite{FiltserFK20}. Their data structure uses space in $ n \cdot \Oh(1/\eps)^{kd} + \Oh(nm)$ and query time in $\Oh(k d)$, where $k$ denotes the complexity of the query (measured in the number of vertices), $m$ denotes the complexity of an input curve and $n$ denotes the number of input curves. 
It is an interesting question whether the same bounds can be obtained for the continuous Fr\'echet distance.
At first glance, the discrete and continuous variants of the Fr\'echet distance seem very similar, but there is an important difference: while the metric space of bounded complexity curves under the discrete Fr\'echet distance has bounded doubling dimension, this does not hold in the continuous case, even when restricted to polygonal curves of constant complexity~\cite{DKS16}. (A metric space has doubling dimension at most $d$ if any ball of any radius $r$ can be covered by $2^d$ balls of radius $\frac{r}{2}$.) This immediately shows that the technique employed by Filtser et al., which effectively applies a doubling oracle to the metric balls centered at input curves (more specifically, simplifications thereof), does not directly extend to the continuous Fr\'echet distance, since such a doubling oracle cannot exist in this case. 

So the \emph{discrete} Fr\'echet distance has a simple ANN that seems optimal, but there is indication that for the \emph{continuous} Fr\'echet distance resolving the time complexity of ANN is more challenging.

Note that it is possible to reduce the ANN problem for the continuous Fr\'echet distance to the ANN problem for the discrete Fr\'echet distance by subsampling along the continuous curves. However, it seems that this approach introduces an (otherwise avoidable) dependency on the arclength.
In 2018, Driemel and Afshani \cite{AD18} described data structures based on multi-level partition trees (using semi-algebraic range searching techniques) which can also be used for exact near neighbor searching under the continuous Fr\'echet distance. For $n$ curves of complexity $m$ in $\RR^2$, their data structure uses space bounded by $n \cdot (\log \log n)^{\Oh(m^2)}$ and the query time is bounded by $\sqrt{n} \cdot (\log n)^{\Oh(m^2)}$. (If the input is restricted to curves in $\RR$, these bounds can be slightly improved.) 
Recently, Driemel and Psarros~\cite{DP21} obtained bounds for the continuous Fr\'echet distance that are similar to the bounds of Filtser et al., albeit at the expense of a higher approximation factor and only for curves in $\RR$. They present a $(5+\eps)$-ANN data structure which uses space in $n\cdot  \Oh\left({\frac 1 \eps}\right)^{k} + \Oh(nm)$ 
and has query time in $\Oh\left(k\right)$, and  a $(2+\eps)$-ANN data structure, which uses space in 
$n\cdot  \Oh\left(\frac{m}{k\eps}\right)^{k} + \Oh(nm)$ 
and has query time in $\Oh\left(k\cdot 2^k\right)$. 
Even more efficient data structures can be obtained at the expense of an even larger approximation factor, see the work of Driemel, Silvestri, and Psarros~\cite{DS17} and~\cite{DP21} which uses locality-sensitive hashing.
In these results neither the space nor the query time is exponential in the complexity of the curves (neither input nor query), but the approximation factor is linear in the query complexity $k$.

\subsubsection*{(Unconditional) lower bounds} Given these results, one may ask whether the cited bounds are optimal for the respective approximation factor that they guarantee. We review some efforts in answering this question and discuss the limitations of the current techniques.
Driemel and Psarros~\cite{DP21, DP20} approach this question using a technique by Miltersen~\cite{miltersen94} for proving cell-probe lower bounds. Their results indicate that any data structure answering a query for a near neighbor under the continuous Fr\'echet distance by using only a constant number of probes to memory cells cannot have a space usage that is independent of the arclength of the input curves (assuming a query radius of $1$).
In addition, their bounds indicate that, in some cases, space exponential in the complexity of the query $k$ is necessary. 
However, these bounds hold only for data structures that use a constant number of probes to memory cells for answering a query, while we would also be interested in data structures that use higher query time, such as $\Oh(k)$ or $\Oh(\log n)$.
A different lower bound technique was used by Driemel and Afshani \cite{AD18}. They show a lower bound in the pointer model of computation on the space-time tradeoff for \emph{range reporting} under the Fr\'echet distance. In this problem, all curves contained inside the query radius need to be output by the query. The resulting lower bound matches the above cited upper bounds even up to the asymptotic number of factors of $\log(n)$. 
The proof uses a construction of input curves in $\RR^2$ and a set of queries, such that the intersection of any two query results has small volume while the queries themselves have large volume. The main drawback of this technique is that, being a volume argument, it inherently uses the fact that all curves inside the query need to be returned and therefore it cannot easily be applied in the near neighbor setting. 

\subsubsection*{Conditional lower bounds} 
The recent rise of fine-grained complexity has also lead to a renewed interest in conditional lower bounds for nearest neighbor data structure problems, see, e.g.~\cite{AlmanW15,AbboudRW17,Rubinstein18,CGLRR19,DBLP:conf/soda/ChenW19}. These lower bounds are for the offline version of the data structure problem, by considering the total time needed for preprocessing and performing a number of queries. They are obtained in a similar way as NP-hardness, specifically via reductions from some fine-grained hypothesis 
such as the \emph{Strong Exponential Time Hypothesis (SETH)}~\cite{IP01} or the \emph{Orthogonal Vectors Hypothesis (OVH)}~\cite{Williams05}. In the Orthogonal Vectors problem we are given two sets of vectors $A, B \subseteq \{0,1\}^d$ of size $n$ and ask whether there exist two vectors $a \in A, b \in B$ such that $\langle a, b \rangle = 0$. The hypothesis postulates that for any constant $\eps > 0$ there exists a constant $c>0$ such that there is no algorithm solving the Orthogonal Vectors problem in time $\Oh(n^{2-\eps})$ in dimension $d = c \log n$. It should be noted that OVH is at least as believable as SETH, because SETH implies OVH~\cite{Williams05}.
As an example, based on the OV-hardness of bichromatic Euclidean closest pair~\cite{AlmanW15} and reducing via a variant of OV with unbalanced size $|A| \ll |B|$~\cite{AbboudW14}, one can show that for any 
$\eps,\beta > 0$
there is no data structure for Euclidean nearest neighbors on $n$ points in $\RR^d$ with preprocessing time 
$\Oh(n^{\beta})$
and query time 
$\Oh(n^{1-\eps})$, in some dimension $d = c \log n$. 
This rules out any sublinear query time for any data structure with polynomial preprocessing time, unless OVH fails. 

For computing the Fr\'echet distance of two  polygonal curves there is a tight conditional lower bound~\cite{Bringmann14}, also for the one-dimensional case~\cite{BringmannM16,BuchinOS19}. However, thus far, there seems to be no comprehensive study of conditional lower bounds for the corresponding data structure problem. We want to close this gap and show tight bounds for the case of one-dimensional curves. These are similar in spirit to the Euclidean nearest neighbor lower bounds discussed above.


\subsection{Our Results}

\begin{table}[t]
\centering
\def\arraystretch{1.5}
\caption{Known upper bounds and our results. For the discrete case we only cite the best known result. The space complexity is implicitly bounded by the preprocessing time in each case. Our preprocessing time is randomized; the bounds can be derandomized at the cost of a factor $\log n$ in preprocessing and query time (by using search trees instead of perfect hashing).}
\smallskip
\label{table:1}
\begin{tabular}{||c||c| c| c| c||} 
 \hline
 Fr\'echet distance & Approximation & Preprocessing Time  & Query Time & Reference \\ [0.5ex] 
 \hline\hline
 discrete, dD & $(1+\varepsilon)$-ANN  &  
 $nm \cdot\left( \Oh(\frac{1}{\varepsilon})^{dk} + \Oh( d\log m) \right)$ 
 & $\Oh(dk)$ & \cite{FiltserFK20} \\
 \hline\hline
 continuous, 1D & $(2+\varepsilon)$-ANN & $n \cdot  \Oh(\frac{m}{k\varepsilon})^k$ & $\Oh(1)^k$ & \cite{DP21} \\
  & $(5+\varepsilon)$-ANN & 
 $n \cdot \Oh(\frac{1}{\varepsilon})^{k} + \Oh(n m)$ 
  & $\Oh(k)$ & \cite{DP21} \\
  \hline\hline
 & $(1+\varepsilon)$-ANN & $n \cdot  \Oh(\frac{m}{k\varepsilon})^k$ & $\Oh(1)^k$ & Theorem~\ref{thm:onepluseps} \\ 
 & $(2+\varepsilon)$-ANN & $n  \cdot \Oh(\frac{m}{k\varepsilon})^k$ & $\Oh(k)$ & Theorem~\ref{thm:twopluseps_one} \\
 continuous, 1D & $(2+\varepsilon)$-ANN & 
  $n \cdot \Oh(\frac{1}{\varepsilon})^{k} + \Oh(n m)$ 
& $\Oh(1)^k$ 
 & Theorem~\ref{thm:twopluseps_two} \\
 & $(2+\varepsilon)$-ANN & 
 $\Oh(nm)$ & $\Oh(\frac{1}{\varepsilon})^{k+2}$ 
 & Theorem~\ref{thm:twopluseps_three} \\
 & $(3+\varepsilon)$-ANN & 
  $n \cdot \Oh(\frac{1}{\varepsilon})^{k} + \Oh(n m)$ 
 & $\Oh(k)$ & Theorem~\ref{thm:threepluseps} \\ [1ex] 
 \hline
\end{tabular}

\caption{Our conditional lower bounds. Each row gives an approximation ratio and a setting of $k$ and $m$ where any $\poly(n)$ preprocessing time and $\Oh(n^{1-\eps'})$ query time cannot be achieved simultaneously. The constants $\eps,\eps',c$ are quantified as $\forall \eps, \eps'>0\colon \exists c>0$. By $f(n) \ll g(n)$ we mean $f(n) = o(g(n))$.
We refer to the respective theorems in Section~\ref{section:lowerbounds} for the exact statements.}
\smallskip
\label{table:2}
\begin{tabular}{||c||c|c|c|c|c||} 
 \hline
 Fr\'echet dist. & Approx. & Preproc. & Query & Parameter Setting & Reference \\ [0.5ex] 
 \hline\hline
continuous, 1D & $2-\varepsilon$ &  $\poly(n)$  & $\Oh(n^{1-\eps'})$ & $1 \ll k \ll \log n$ and $m = k \cdot n^{c/k}$ & Thm.~\ref{thm:2minusepshard1d} \\
& $3-\varepsilon$ & $\poly(n)$ & $\Oh(n^{1-\eps'})$ & $m=k=c \log n$ & Thm.~\ref{thm:3minuseps_lb_1d} \\ [1ex] 
\hline
continuous, 2D & $3-\varepsilon$ & $\poly(n)$  & $\Oh(n^{1-\eps'})$ & $1 \ll k \ll \log n$ and $m = k \cdot n^{c/k}$ & Thm.~\ref{thm:3minuseps_lb_2d} \\ [1ex] 
 \hline
\end{tabular}
\end{table}

For the discrete Fr\'echet distance the ANN problem is by now well understood, but the continuous Fr\'echet distance remains very challenging. Therefore, in this paper we focus on the important special case of one-dimensional curves, which arise in various domains such as finance and signal processing, where they are typically called ``time series''. We give several new data structure bounds for the problem of approximate near neighbor searching for one-dimensional curves under the continuous Fr\'echet distance. 
Table~\ref{table:1} provides an overview of our upper bounds, compared to known results.
In the second part of our paper, we show that most of these upper bounds are tight under the Orthogonal Vectors Hypothesis, when viewed as offline problems where the input and the set of queries are given in advance. 
To obtain these lower bounds, we introduce a novel OV-hard variant of Orthogonal Vectors in which one set contains sparse vectors, i.e., vectors that only contain few 1s; this problem may be of independent interest.
Table~\ref{table:2} gives an overview of our lower bound results.
To argue that most of our upper bounds are tight, we consider the following general scenario:
\begin{quote}
\emph{Suppose we have an $\alpha$-ANN for some fixed constant $\alpha$, we run its preprocessing on a data set of $n$ curves, and then we run $n$ queries.}
\end{quote}
In particular, consider this scenario for the following three ranges of $\alpha$.
\begin{itemize}
\item $1 < \alpha < 2$: Using our $(1+\eps)$-ANN, this scenario takes total time $n \cdot \Oh(\frac{m}{k\varepsilon})^k$, which simplifies to $n \cdot \Oh(\frac{m}{k})^k$ since $\eps = \alpha-1$ is fixed. Assuming OVH, our first lower bound shows that this running time cannot be improved to $n \cdot f(k) \cdot (\frac m k)^{o(k)}$ for any function $f$, for the following reason. Pick $k = k(n)$ sufficiently small such that $f(k) = n^{o(1)}$. Pick $m = k \cdot n^{c/k}$, so that $(\frac m k)^{o(k)} = (\frac{k \cdot n^{c/k}}{k})^{o(k)} = n^{o(1)}$. Then the total running time would be $n \cdot f(k) \cdot (\frac m k)^{o(k)} = n^{1+o(1)}$, which contradicts that either the preprocessing time is superpolynomial or the query time near-linear, as stated in Theorem~\ref{thm:2minusepshard1d}. This shows that the factor $(\frac m k)^{\Theta(k)}$ in our running time is necessary.
Our second lower bound shows that the running time cannot be improved to $n \cdot (\frac m k)^{f(k)} \cdot 2^{o(k)}$ for any function~$f$, as for $m = k = c \log  n$ the total time would become $n \cdot (\frac m k)^{f(k)} \cdot 2^{o(k)} = n \cdot 1^{f(k)} \cdot n^{o(1)} = n^{1+o(1)}$, which contradicts that either the preprocessing time is superpolynomial or the query time near-linear, as stated in Theorem~\ref{thm:3minuseps_lb_1d}. This shows that the factor $\Oh(1)^k$ in our query time is necessary. In this sense, the running time of our $(1+\eps)$-ANN is tight.
    
\item $2 < \alpha < 3$: By using our second or third $(2+\eps)$-ANN (Theorem~\ref{thm:twopluseps_two} or~\ref{thm:twopluseps_three}) we solve this scenario in total time $\Oh(nm) + n \cdot \Oh(\frac 1 \eps)^{k+2}$, which simplifies to $\Oh(nm) + n \cdot \Oh(1)^k$ since $\eps = \alpha-2$ is fixed. Assuming OVH, our second lower bound shows that this cannot be improved to time $n \cdot (\frac m k)^{f(k)} \cdot 2^{o(k)}$ for any function~$f$, as for $m = k = c \log  n$ we would obtain a total time of $n \cdot (\frac m k)^{f(k)} \cdot 2^{o(k)} = n \cdot 1^{f(k)} \cdot n^{o(1)} = n^{1+o(1)}$, which contradicts that either the preprocessing time is superpolynomial or the query time near-linear, as stated in Theorem~\ref{thm:3minuseps_lb_1d}. This shows that the factor $\Oh(1)^k$ in our running time is necessary. In this sense, the running time of our $(2+\eps)$-ANNs from Theorems~\ref{thm:twopluseps_two} and~\ref{thm:twopluseps_three} are tight. (Our $(2+\eps)$-ANN from Theorem~\ref{thm:twopluseps_one} is not tight in this sense, but it realizes a different tradeoff between preprocessing and query time.)

\item $\alpha>3$: In this range, our ANNs still require exponential time in terms of $k$, but we cannot hope for a tight lower bound using the current techniques. This is due to a fundamental limitation of proving inapproximability factor $>3$ for a metric problem, cf.~e.g.~\cite[Open Question 3]{Rubinstein18}. For this reason, we have no tight lower bounds in this range.
\end{itemize}

\subsection{Technical Overview}\label{section:overview}

The high-level view of our data structures employs a well-known technique: exhaustively enumerate a strategic subset of the query space with a set of ``candidate'' query curves during preprocessing, and store the answers to these candidate queries in a dictionary. During query time, we apply a simple transformation to the query curve (such as rounding vertices to a scaled integer grid) and look up the answer in the dictionary. Filtser et al.~\cite{FiltserFK20}  used this technique for the discrete Fr\'echet distance and Driemel and Psarros~\cite{DP21} showed that it can also be applied for the continuous Fr\'echet distance of one-dimensional curves. A particular challenge that appears in the continuous case is that the doubling dimension can be unbounded, even if the complexity of the curves is small. Intuitively, what can happen is that the query contains some small noise that appears in the middle of a long edge. The continuous Fr\'echet distance---being robust to this noise---may match these short edges to the interior of a long edge on the near neighbor input curve. However, we cannot afford to generate all possible noisy query curves of this type, since this would introduce a dependency on the arclength in our time and space bounds. 
Driemel and Psarros overcome this challenge with the use of signatures, which allow to ``guess'' the approximate shape of a query curve within some approximation factor. The idea is that the signature acts as a ``low-pass'' filter that eliminates the noisy short edges. However, this is a delicate process as the signature may eliminate too many edges on one of the curves (either on the near neighbor or on the query curve) leading to the near neighbor being missed during query time. In addition, the process may introduce false-positives, hence the high approximation factor of $(5+\eps)$ in the result of~\cite{DP21}. 

We see our contributions as three-fold:
\begin{enumerate}
\item Our first contribution is to improve the approximation factors of Driemel and Psarros~\cite{DP21} while staying within the same time bounds, cf.~Table~\ref{table:1} for a comparison.
\begin{enumerate}
\item For Theorem~\ref{thm:threepluseps}, we use almost the same algorithm as Driemel and Psarros, but combine this with a more careful analysis based on new observations on the Fr\'echet distance of approximately monotone curves. As a result, we can achieve a $(3+\eps)$-approximation within the same time bounds as the previous $(5+\eps)$-ANN.
\item In Theorem~\ref{thm:onepluseps} we even achieve an approximation factor of $(1+\eps)$ within the same time bounds as the previous $(2+\eps)$-ANN. To achieve this result, we introduce the concept of \straightenings{} in Section~\ref{section:lemmas}. \Straightenings{} share some properties of signatures, but they provide a more refined approximation, leading to fewer false positives. They allow us to ``guess'' the shape of a query curve up to approximation factor $(1+\eps)$.
\end{enumerate}
We derive useful properties of both signatures and \straightenings{}. Central to our analysis is the concept of $\delta$-visiting orders, which we introduce in Section~\ref{section:lemmas} and analyze in Section~\ref{section:missingproofs}. 

\item Our second contribution is a range of data structures for the $(2+\eps)$-ANN which together provide a tradeoff between preprocessing time and query time (see Theorems~\ref{thm:twopluseps_one}, \ref{thm:twopluseps_two}, and~\ref{thm:twopluseps_three}). In each case, the preprocessing time implicitly bounds the number of candidates that are generated and therefore the size of the dictionary used by the data structure. Thus, these data structures also achieve a tradeoff between space and query time. An important observation that leads to this result is that the enumeration of candidates can be ``dualized'' and then be shifted from the preprocessing time to the query time. In the extreme case, this allows us to design a data structure that has linear preprocessing time and space, by performing most of the candidate generation during query time, see Theorem~\ref{thm:twopluseps_three} for the exact result.

\item Given the diverse range of upper bounds, it is natural to ask if these bounds can be improved. Our third main contribution is to show that most of our upper bounds are tight under the Orthogonal Vectors Hypothesis.
All known OV-based hardness results for the Fréchet distance encode each of the dimensions using at least one vertex, thus transforming $d$-dimensional vectors into curves of length $k=\Omega(d)$. Since OVH postulates a lower bound in dimension $d = c \log n$, it is thus natural to prove OV-based lower bounds for curves of length $k=c \log n$. Our lower bound in Theorem \ref{thm:3minuseps_lb_1d} handles this setting, cf.~Table~\ref{table:2}. 

However, for some of our lower bounds we require $k = o(\log n)$, as this is necessary to rule out time $(m/k)^{o(k)}$.
Surprisingly, we overcome the barrier of using at least one vertex per dimension. Specifically, we prove OV-based lower bounds for any $1 \ll k \ll \log n$, see Theorem~\ref{thm:2minusepshard1d}.
For this, we use two crucial observations: 
(i) it is possible to only encode the 1s of one vector set, while the 0s do not require any additional vertices on the curve, and (ii) we can show hardness of a variant of OV where one set contains only sparse vectors, i.e., vectors with a very small number of 1s. See Theorem \ref{thm:2minusepshard1d} for the hardness result we obtain in this case.
Interestingly, a similar construction is also possible for $(3-\eps)$-ANN for two-dimensional curves, see Theorem \ref{thm:3minuseps_lb_2d}.
\end{enumerate}

\subsubsection*{Organization}
In Section~\ref{section:prelims} we define the notation and state some known facts and observations. In Section~\ref{section:lemmas} we define key concepts, and we present their properties and our main technical lemmas. 
Our data structures are described and analyzed in Sections \ref{section:datastructure1apprx}, \ref{section:datastructure2apprx}, and \ref{section:datastructure3apprx}. 
In Section~\ref{section:missingproofs} we prove our main technical lemmas. 
In Section~\ref{section:lowerbounds} we present our conditional lower bounds. 

\section{Preliminaries}\label{section:prelims}

For any positive integer $n$, we define $[n] \coloneqq \{1,\ldots,n\}$. For any two points $p,q\in \RR^d$, $\overline{pq}$ denotes the directed line segment connecting $p$ with $q$ in the direction from $p$ to $q$.
Any sequence of points $p_1,\ldots,p_m \in \RR^d$ defines a polygonal curve formed by the ordered line segments $\overline{p_i p_{i+1}}$. We call the points $p_i$ the \emph{vertices} of the curve and the line segments $\overline{p_i p_{i+1}}$ the \emph{edges}. The resulting curve can be viewed as a continuous function $P\colon [0,1]\mapsto \RR^d$. For $d=1$, we may refer to the curve as a \emph{one-dimensional curve} or as a \emph{time series}. We define the \emph{complexity} of a polygonal curve $P$ as the number of its vertices and we denoted it by $|P|$.
We say a polygonal curve is \emph{degenerate} if there are three consecutive vertices $p,q,r$, such that $q$ lies on the line segment $\overline{p r}$.  In this case, we call $q$ a degenerate vertex of this curve.
Given a sequence of points $p_1,\ldots,p_m$, we can define a non-degenerate curve by omitting degenerate vertices. We denote the resulting curve by $\seqtocurve{ p_1,\ldots,p_m}$. Note that for one-dimensional curves, the vertices of the resulting non-degenerate curve are the extrema of the function.  
For any two $0\leq t_a<t_b\leq1$ and any curve $P$, we denote by $P[t_a,t_b]$ the subcurve of $P$ starting at $P(t_a)$ and ending at $P(t_b)$. For any two curves $P$, $Q$, with vertices $p_1,\ldots,p_a$ and $p_b,\ldots,p_m$, respectively, 
$P \circ Q $ denotes the polygonal curve $\seqtocurve{ p_1,\ldots,p_a,p_b, \ldots p_m } $, that is the \emph{concatenation} of $P$ and $Q$. For $n$ polygonal curves $P_1,\ldots,P_n$, we denote by $\bigcirc_{i=1}^n P_i$ the concatenation $P_1\circ P_2 \circ \cdots \circ P_n$. 
Given a polygonal curve $P = \seqtocurve{p_1, \dots, p_m}$ and a point $x$ in $\RR^d$, we define the translated curve as $P + x \coloneqq \seqtocurve{p_1 + x, \dots, p_m + x}$.
For a point $x\in \RR^d$ and a polygonal curve $P$, we use the notation $x \in P$ to indicate that there exists a $t\in[0,1]$ such that $P(t)=x$. Let $\GGG_{\eps}:=\{i\cdot \eps \mid i\in \ZZ\}$ be the regular grid with side-length $\eps>0$.

We will use the following known observations~(see also \cite{buchin2008computing} and \cite{driemel2013jaywalking}).

\begin{observation}\label{observation:linesegment}
For any two line segments $X=\overline{ab}$, $Y=\overline{cd}$ it holds that
$\df(X,Y)=\max\{\|a-c\|,\|b-d\|\}$.
\end{observation}

\begin{observation}
\label{observation:concatenation}
Let two polygonal curves $Q : [0, 1]\mapsto \RR^d$ and
$P : [0, 1] \mapsto \RR^d$ be the concatenations of two subcurves
each, $Q = Q_1\circ Q_2$ and $P= P_1\circ P_2$. Then it holds that
$\df (P, Q) \leq \max \{\df (Q_1, P_1), \df (Q_2, P_2)\}$.
\end{observation}

\begin{observation}\label{observation:shortcut}
Let $Q$ be a line segment and let $P$ be a curve with $\df(P,Q) \leq \delta$. Let $P'$ be a curve that is formed from a subsequence of the vertex sequence of $P$ including the first and last vertex of $P$. Then, $\df(P',Q) \leq \delta$.
\end{observation}

We also make use of an algorithm by Alt and Godau~\cite{AltG95} for deciding whether the Fr\'echet distance between two polygonal curves exceeds a given threshold. 
 \begin{theorem}[\cite{AltG95}]\label{theorem:frechetdecision}
 There is an algorithm which, given polygonal curves $P$, $Q$ and a threshold parameter $\delta>0$, decides in $\Oh(|P|\cdot |Q|)$ time whether $\df(P,Q)\leq \delta$.
 \end{theorem}

Our data structures can be implemented to work on the Word-RAM and under certain assumptions on the Real-RAM, as discussed next. Central to our approach is the use of a dictionary, which we define as follows.


\begin{Definition}[Dictionary]
A \emph{dictionary} is a data structure which  stores a set of (key, value) pairs and when presented with a key, either returns the corresponding value, or returns that the key is not stored in the dictionary.
\end{Definition}

In the Word-RAM model, such a dictionary can be implemented using perfect hashing. For storing a set of $n$ (key,value) pairs, where the keys come from a universe $U^k$,  perfect hashing provides us with a dictionary using $\Oh(n)$ space and $\Oh(k)$ query time which can be constructed in  $\Oh(n)$ expected time~\cite{FKS84}.   During look-up, we compute the hash function in $\Oh(k)$ time, we access the corresponding bucket in the hashtable in $\Oh(1)$ time and check if the key stored there is equal to the query in $\Oh(k)$ time. This gives an efficient randomized implementation of dictionaries. Alternatively, we can use balanced binary search trees and pay an additional $\log n$ factor in preprocessing and query time of the dictionary. This deterministic algorithm also works in the Real-RAM model, if we assume that the floor function can be computed in constant time---a model which is often used in the literature~\cite{H11}.   
In the Word-RAM model, we use the standard assumption that the word size is logarithmic in the size of the input, and we ensure that all numbers (vertices of the time series, results of intermediate computations, etc.) are restricted to be of the form $a/b$ where $a$ is an integer in $[-(nm)^{\Oh(1)},(nm)^{\Oh(1)}]$ and $b=(nm)^{\Oh(1)}$. 

\section{Simplifications, signatures, and \straightenings}\label{section:lemmas}


In this section we state the main definitions and lemmas. 
To allow for an easier understanding of our results, we then already describe our algorithms and prove correctness using these lemmas. In Section \ref{section:missingproofs} we then give the proofs of the lemmas presented in the current section.

\subsection{Definitions}

Let us start with two basic definitions.

\begin{Definition}
We say a curve $P:[0,1]\rightarrow \RR$ is \emph{\deltamonotone{}} if one of the following statements holds:
\begin{compactenum}[(i)]
    \item $\forall~ t<t'\in [0,1]:$ $P(t')\geq P(t) -\delta $,
    \item $\forall~ t<t'\in [0,1]:$ $P(t')\leq P(t) +\delta $. 
\end{compactenum}
More specifically, we say the curve is $\delta$-\emph{monotone increasing} in case (i) and $\delta$-\emph{monotone decreasing} in case (ii). Note that a curve can be both $\delta$-monotone increasing and decreasing at the same time. In addition, we may say  $P$ is $\delta$-monotone with respect to a directed edge $\overline{ab}$, if $a \leq b$ in case (i) and if $b \leq a$ in case (ii).
\end{Definition}

\begin{Definition}
The $\delta$-range of a point $p \in \RR$ is the interval $B(p,\delta)=[p-\delta,p+\delta]$. The $\delta$-range of a curve $P$ is the interval $B(P,\delta)=\bigcup_{x\in P} B(x,\delta)$. 
\end{Definition}

We now define the notion of \emph{simplification} that we use in this work.

\begin{Definition}[$\delta$-simplification]
Given a curve $ P:~ [0, 1] \mapsto \RR^d$, a $\delta$-simplification is a curve $P' :~ [0, 1] \mapsto \RR^d$ that is given as $P'= \langle P(t_1),\dots,P(t_{\ell}) \rangle$ for a sequence of values $0=t_1 < \dots < t_{\ell}=1$, such that each $P(t_i)$ is a vertex of $P$,
$P'$ is non-degenerate, and
\begin{equation}\label{locality-property}
\df(\overline{P(t_i) P(t_{i+1})}, P[t_i,t_{i+1}]) \leq \delta,\ \text{for all } 1 \leq i < \ell.
\end{equation}
\end{Definition}

We also refer to (\ref{locality-property}) as the \emph{locality property}. Furthermore, note that if $P'$ is a $\delta$-simplification of $P$, then $\df(P,P') \leq \delta$ and the complexity of $P'$ is at most the complexity of $P$.
Note that the vertices of a $\delta$-simplification $P'$ give us a natural partition of $P$. Furthermore, we want to highlight that our definition of a simplification is one out of many definitions that are used in literature. In particular, in other work curves which are degenerate or non vertex-restricted are also called simplifications.
Now we define some properties that a simplification can or must have.

\begin{observation}[direction-preserving property]
For any  $\delta$-simplification $P' = \langle P(t_1),\dots,P(t_{\ell}) \rangle$ of a curve  $ P:~ [0, 1] \mapsto \RR$ and any index $i$, the subcurve $P[t_i,t_{i+1}]$ is $2\delta$-monotone with respect to $\overline{P(t_i)P(t_{i+1})}$.
\end{observation}

\begin{Definition}[vertex-range-preserving property]
Let $P' = \langle P(t_1),\dots,P(t_{\ell}) \rangle$ be a $\delta$-simplification of a  curve $P:~[0,1]\mapsto \RR$.
We say $P'$ is range-preserving on the vertex $P(t_i)$ if the following holds:
\begin{enumerate}[(i)]
\item if $P(t_i)$ is a local maximum on $P'$, then $P(t) \leq P(t_i)$ for all $t$ in $[t_{i-1},t_{i+1}]$, and
\item if $P(t_i)$ is a local minimum on $P'$, then $P(t) \geq P(t_i)$ for all $t$ in $[t_{i-1},t_{i+1}]$.
\end{enumerate}
We say $P'$ is vertex-range-preserving, if it is vertex-range-preserving on all interior vertices.
\end{Definition}

\begin{Definition}[edge-range-preserving property]
Let $P' = \langle P(t_1),\dots,P(t_{\ell}) \rangle$ be a  $\delta$-simplification of $P:~[0,1]\mapsto \RR$. 
We say that $P'$ is edge-range-preserving on edge $\overline{P(t_i)P(t_{i+1})}$ if for any $t \in [t_i, t_{i+1}]$ it holds that $P(t) \in \overline{P(t_i)P(t_{i+1})}$. We say $P'$ is edge-range-preserving if this condition holds for all edges of $P'$.
\end{Definition}

Note that the vertex-range-preserving property is implied by the edge-range-preserving property, but not the other way around. However, the vertex-range preserving property implies the edge-range-preserving property on all edges except the first and the last edge. 

\begin{Definition}[$\delta$-edge-length property]
We say that a one-dimensional curve $P = \langle p_1, \dots, p_m \rangle$ has the $\delta$-edge-length property if
\begin{itemize}
    \item $|p_1 - p_2| > \delta$ and $|p_{m-1} - p_m| > \delta$, and
    \item $|p_i - p_{i+1}| > 2\delta$ for all $i \in \{2,\ldots,m-2\}$.
\end{itemize}
\end{Definition}

Finally, we can define two of the main concepts that we use in our algorithms: $\delta$-signatures and $\delta$-\straightenings. These two definitions help us to preprocess the input set of one-dimensional curves and the query curve in ways such that an efficient retrieval is possible.

\begin{Definition}[$\delta$-signature]
A  $\delta$-simplification $P'$ of a one-dimensional curve $P$ is a $\delta$-signature if it has the $\delta$-edge length property and is vertex-range-preserving. 
\end{Definition}

\begin{Definition}[$\delta$-\straightening]
A $\delta$-simplification $P'$ of a one-dimensional curve $P$ is a $\delta$-\straightening if it is edge-range-preserving.
\end{Definition}

The above definition of a $\delta$-signature is equivalent to the definition given in~\cite{DKS16}. 
For any $\delta>0$ and any curve $P:~[0,1]\mapsto \RR$ of complexity $m$, a $\delta$-signature of $P$ can be computed in $\Oh(m)$ time \cite{DKS16}. 
The $\delta$-signature of a   curve is unique under certain general-position assumptions, however we do not explicitly use this property in our proofs.
Note that $\delta$-\straightenings are \emph{not} unique. In fact, there can be many different $\delta$-\straightenings of the same curve, e.g., $P$ itself is a $\delta$-\straightening of $P$ for any $\delta > 0$. 
We give an example of a signature and different \straightenings of the same curve in Figure~\ref{fig:defexamples}.

\begin{figure}
    \centering
	\includegraphics[scale=1.1]{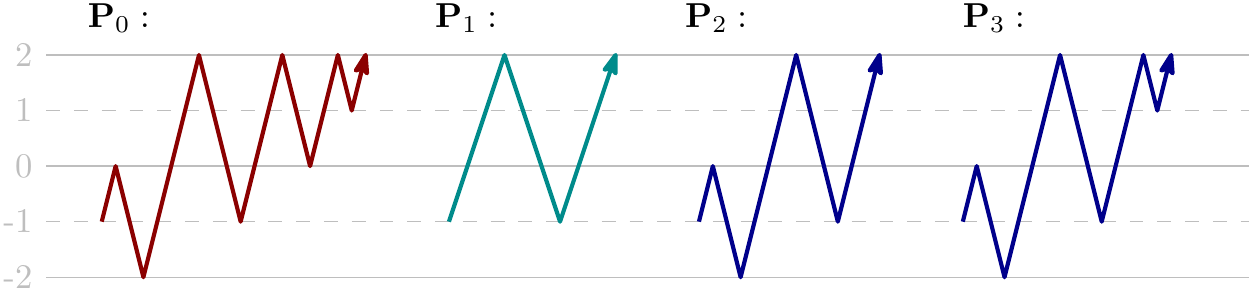}
    \caption{$P_1$ is a $1$-signature of $P_0$, whereas $P_2$ and $P_3$ are  $1$-\straightenings of $P_0$.}
    \label{fig:defexamples}
\end{figure}


We introduce the notion of visiting orders, which we will use to prove correctness of our data structures.   
\begin{Definition}\label{def:visitingorder}
Let $P:[0,1] \rightarrow \RR$ and $Q:[0,1] \rightarrow \RR$ be curves. Let $u_1,\dots,u_{\ell}$ denote the ordered vertices of $Q$ and let $v_1,\dots,v_{m}$ denote the ordered vertices of $P$. A (partial) \textbf{$\delta$-visiting order} of $Q$ on $P$ is a sequence of indices $i_1 \leq \dots \leq i_{\ell}$, such that $|u_j - v_{i_{j}}| \leq \delta$ for each vertex $u_{j}$ of $Q$. 
\end{Definition}
In particular, if we know that there exists a $\delta$-visiting order of $Q$ on $P$, then we can approximately ``guess'' $Q$ from the vertex sequence of $P$, by enumerating all possible visiting orders of the vertices of $P$ and for any fixed visiting order, enumerating all eligible grid sequences within the $\delta$-ranges of these vertices.

Driemel, Krivosija and Sohler proved the following lemma (rephrased using $\delta$-visiting orders).
\begin{lemma}[Lemma 3.2 \cite{DKS16}]\label{lemma:signatures2} 
Let $P:[0,1] \rightarrow \RR$ and $Q:[0,1] \rightarrow \RR$ be curves and let $P'$ be a $\delta$-signature of $P$. If $\df (P,Q) \leq \delta$, then there exists a $\delta$-visiting order of $P'$ on $Q$. 
\end{lemma}


\subsection{Main lemmas}

In this section we present the main lemmas for signatures and \straightenings that we will use in Sections~\ref{section:datastructure1apprx}~to~\ref{section:datastructure3apprx}. Their proofs are deferred to Section~\ref{section:missingproofs}.

Most of our lemmas improve the basic triangle inequality $\df(P,Q) \le \df(P,X) + \df(X,Q)$ in some situations involving signatures and straightenings.

\begin{restatable}{lemma}{lemmastraightenings}\label{lemma:simplproxy}\label{lemmastraightenings}

 Let $P:[0,1]\mapsto \RR$ and $Q:[0,1]\mapsto \RR$ be two curves and let $Q'$ be any $\delta$-\straightening of $Q$. If $\df(P,Q') \leq  \delta$ then $\df(P,Q)\leq \delta$.
\end{restatable}

We would like to show the equivalent statement of Lemma~\ref{lemma:simplproxy} for signatures. However, as the  example in Figure~\ref{fig:counterexample} shows, this is not possible. Instead, we show a slightly weaker bound in the following lemma.

\begin{restatable}{lemma}{lemmasignatureproxy}\label{lemma:signatureproxy}\label{lemmasignatureproxy}

 Let $\delta=\delta'+\delta''$ for $\delta,\delta',\delta''\geq 0$ and let $P:[0,1]\mapsto \RR$ and $Q:[0,1]\mapsto \RR$ be two curves. Let $Q'$ be any $\delta'$-signature of $Q$. If $\df(Q',P) \leq \delta$, $|Q(0)-P(0)| \leq \delta''$, and $|Q(1)-P(1)| \leq \delta''$, then $\df(P,Q)\leq \delta$.
\end{restatable}

Note that Lemma~\ref{lemma:simplproxy} is much stronger than what we would get by merely applying the triangle inequality on the Fr\'echet distances on the curves $P$, $Q$ and $Q'$. Lemma~\ref{lemma:signatureproxy}, although weaker, is still stronger than the bound we would get from the triangle inequality.
To illustrate this we include the following corollary. Note that merely using triangle inequality would yield $\df(P,Q) \le 6\delta$, instead of $\df(P,Q) \le 3\delta$. 

\begin{corollary}
  For one-dimensional curves $P,Q$ let $P'$ be a $\delta$-signature of $P$, and let $Q'$ be the $2\delta$-signature of $Q$. If $\df(P',Q') \le 3\delta$ and $|P'(0)-Q'(0)| \le \delta, |P'(1)-Q'(1)| \le \delta$, then $\df(P,Q) \le 3\delta$.
\end{corollary}
\begin{proof}
  Follows from applying of Lemma~\ref{lemma:signatureproxy} twice.
  We first apply the lemma to $P'$, $Q'$ and $P$ and obtain $\df(P,Q')\leq 3\delta$. In the second step, we apply the lemma to $P$, $Q'$ and $Q$ and obtain $\df(P,Q)\leq 3\delta$. 
\end{proof}

\begin{figure}
    \centering
	\includegraphics[scale=1.1]{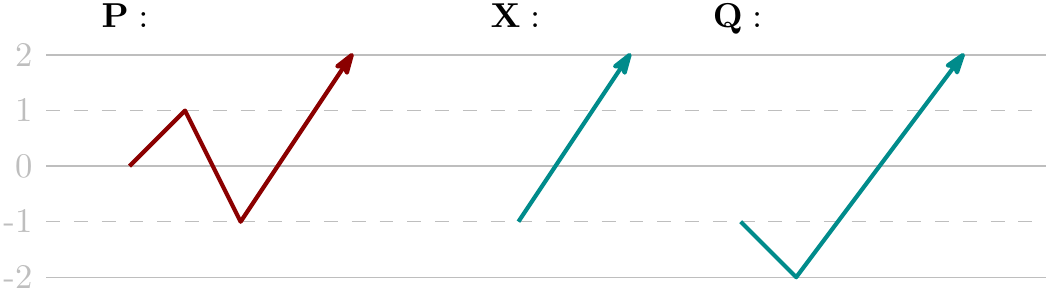}
    \caption{This example shows that an equivalent statement of Lemma~\ref{lemma:simplproxy} for signatures is not true. The curve $X=\seqtocurve{-1,2}$ is a $1$-signature of $Q=\seqtocurve{-1,-2,2}$ and the curve $P=\seqtocurve{0,1,-1,2}$ has Fr\'echet distance $1$ to $X$, but the Fr\'echet distance of $P$ to $Q$ is $2$.}
    \label{fig:counterexample}
\end{figure}


The following lemma is used to show correctness for our $(1+\eps)$ and $(2+\eps)$-ANN.

\begin{restatable}{lemma}{lemmagoodsimplificationranges}
\label{lemma:goodsimplificationranges2}
 Let $P:[0,1]\mapsto \RR$ and $Q:[0,1]\mapsto \RR$ be curves such that $\df(Q,P)\leq \delta$, there exists a $\delta$-\straightening  $Q'$ of $Q$ which satisfies the following properties:
 \begin{compactenum}[(i)]
     \item there exists a $11\delta$-visiting order of $Q'$ on $P$, and
     \item $\df(Q',P)\leq \delta$. 
 \end{compactenum}
\end{restatable}

We use the following lemma to show correctness for our $(3+\eps)$-ANN. One part of the lemma statement, the existence of a $2\delta$-visiting orders, was already used in~\cite{DP20}. However, the resulting approximation factor of the ANN obtained there was $(5+\eps)$. In order to show correctness of our $(3+\eps)$-ANN, it is necessary to prove the bound of $3\delta$ on the resulting Fr\'echet distance of the two signature curves. Note that the triangle inequality implies a bound of $4\delta$---which would not be sufficient for us.

\begin{restatable}{lemma}{lemmasignatures}
\label{lemma:signatures3}
  For one-dimensional curves $P,Q$ let $P'$ be a $\delta$-signature of $P$, and let $Q'$ be a $2\delta$-signature of $Q$. If $\df(P,Q) \le \delta$ then $\df(P',Q') \le 3\delta$ and there exists a $2\delta$-visiting order of $Q'$ on~$P'$.
\end{restatable}

\section{$(1+\eps)$-Approximation}
\label{section:datastructure1apprx}

In this section, we show that there exists a $(1+\eps)$-ANN data structure for  one-dimensional curves under the Fr\'echet distance, with space in $n\cdot \Oh(\frac{m}{k\eps})^k$, expected preprocessing time in  $nm\cdot \Oh(\frac{m}{k\eps})^k$ and query time in $\Oh(k \cdot 2^k)$. We describe the data structure in Section~\ref{subsection:datastructure1apprx_ds} and we analyze its performance in Section~\ref{subsection:datastructure1apprx_an}.
\subsection{The data structure}
\label{subsection:datastructure1apprx_ds}
\subsubsection*{Data structure}
We are given as input a set of one-dimensional curves $\mathcal{P}$, as sequences of vertices, the distance threshold $\radius>0$, the approximation error $\eps>0$, and the complexity of the supported queries $k$. 
To discretize the query space, we use the grid  $\GGG_{\eps\radius/2}$ (recall that $\GGG_{\eps}:=\{i\cdot \eps \mid i\in \ZZ\}$ is the regular grid with side-length $\eps$). Let $\HHH$ be a dictionary which is initially empty. 
For each input one-dimensional curve $P \in \mathcal{P}$  
we compute a set $\CCC':=\CCC'(P)$ which contains all curves $Q$ such that: 
\begin{inparaenum}[i)]
\item $Q$ has complexity at most $k$,
\item all vertices of $Q$ belong to $\GGG_{\eps\radius/2}$, and
\item there is an $((11+\eps/2)\radius)$-visiting order of $Q$ on $P$.
\end{inparaenum}
%
Formally, 
\begin{multline*}
\CCC'=\{ \seqtocurve{u_1,\ldots,u_{\ell}}   \mid  \ell \leq k \text{ and }\\ \exists (i_1,\ldots,i_{\ell}) ( i_1\leq \dots \leq i_{\ell}  \text{ and } (\forall j\in [\ell]) (u_j \in B(p_{i_j},(11+\eps/2)\radius)\cap \GGG_{{\eps\radius}/{2}})) \}.
\end{multline*}

 Next, we filter $\CCC'$ to obtain 
the set  $\CCC(P)=\{Q\in \CCC' \mid \df(Q,P)\leq (1+\eps/2)\radius \}$. We store $\CCC(P)$ in $\HHH$ as follows: for each $Q\in \CCC(P)$, 
if $Q$ is not already stored in $\HHH$, then we insert $Q$
into $\HHH$, associated with a pointer to $P$.  

The complete pseudocode for the preprocessing algorithm can be found in Algorithms~\ref{alg:generateorders} and \ref{alg:preprocessing1apprx}. To achieve approximation factor $(1+\eps)$,  we run \texttt{preprocess}$(P,\radius,\eps/2,k)$.

\subsubsection*{Query algorithm}
Let $Q$ be the query curve with vertices $q_1,\ldots,q_k$ and let $\eps>0$ be the approximation error. The query algorithm first enumerates  all curves $Q'$ such that
\[
Q' \in \{\seqtocurve{q_1,S,q_k} \mid \text{$S$ is a subsequence of $q_2,\ldots,q_{k-1}$}  \}.
\]
For each such $Q'$ we test whether it is a $\radius$-\straightening of $Q$.  To this end, we first test if each shortcut taken in $Q'$ is within distance $\radius$ from the corresponding subcurve of $Q$. Then we check for each shortcut if the corresponding subcurve of $Q$ stays within range by testing all vertices of the subcurve one by one. 
If $Q'$ is a  $\radius$-\straightening of $Q$, then we snap the vertices of $Q'$ to $\GGG_{\eps\radius/2}$,  to obtain a new curve  $Q''$ and we 
 probe $\HHH$: if $Q''$ is stored in $\HHH$, then we return its associated input curve $P\in \mathcal{P}$. If $Q''$ is not stored in $\HHH$, then we return “no”. 

The complete pseudocode for the query algorithm can be found in 
Algorithm~\ref{alg:query1apprx}.  To achieve approximation factor $(1+\eps)$, we run \texttt{query}$(Q,\radius,\eps/2)$.

\label{section:pseudocode1apprx}
\begin{algorithm}
\caption{A call to \texttt{generate\_orders}($m$, $k$) 
returns all $(i_1,\ldots,i_{\ell}) \in [m]^\ell$, where $\ell\in [k]$ and such that $1=i_1\leq \cdots \leq i_{\ell}=m$. We assume $k \ge 2$.\label{alg:generateorders}}

\begin{algorithmic}[1]

\Procedure{\texttt{generate\_orders}}{$m \in \NN$, $k \in \NN$} 
\State $\III_2\gets \{(1,m)\}$ 
\For{\textbf{each} $\ell = 3,\ldots, k$}
\State $\III_{\ell}\gets \emptyset$
\For{\textbf{each} $(i_1,\ldots,i_{\ell-1})\in \III_{\ell-1}$} \Comment{$i_{\ell-1}=m$}
\For{\textbf{each} $j=i_{\ell-2},\ldots,m$}
\State \label{alg:generate_orders:update} $\III_{\ell} \gets \III_{\ell} \cup 
\{(i_1,\ldots,i_{\ell-2},j,m)\}$
\EndFor
\EndFor
\EndFor
\State \Return $\bigcup_{2 \le \ell \le k} \III_{\ell}$
\EndProcedure
\end{algorithmic}
\end{algorithm}

\begin{algorithm}
\caption{Preprocessing algorithm. We call \texttt{preprocess} to build the data structure. \label{alg:preprocessing1apprx}}
\begin{algorithmic}[1]
\Procedure{\texttt{preprocess}}{input set $\PPP$, $\radius>0$, $\eps>0$, $k \in \NN$}
     \State Initialize empty dictionary $\HHH$ 
     \For {{\bf each} $P \in \PPP$}
     \State $\CCC(P)\gets$ \texttt{generate\_keys}$(P, \radius , \eps, k)$
     \For{{\bf each} $Q\in \CCC(P)$}
      \If{$Q$ not in $\HHH$}
     \State insert key $Q$ in $\HHH$, associated with a pointer to $P$  
     \EndIf
     \EndFor
     \EndFor
     \EndProcedure
\Procedure{\texttt{generate\_keys}}{curve $P$, $\radius>0$, $\eps>0$, $k \in \NN$}
\State $\CCC'\gets$\texttt{generate\_candidates}$(P,\radius, (11+\eps),\eps,k)$
\State $\CCC  \gets \emptyset$
\For{{\bf each} $Q\in \CCC'$}
\If{$\df( P ,   {Q} )\leq (1+\eps)\radius$} 
\State $\CCC \gets \CCC \cup \{ Q\}$
\EndIf
\EndFor
\State {\bf  return} $\CCC $
\EndProcedure
\Procedure{\texttt{generate\_candidates}}{curve $P$ with vertices $p_1,\ldots,p_m $, $\radius>0$, $r>0$, $\eps>0$, $k \in \NN$}
\State $\mathcal{S}\gets \emptyset$, $\CCC'\gets \emptyset$
\State $\mathcal{I}\gets$\texttt{generate\_orders}$(m,k)$
\For{\textbf{each} $(i_1,\ldots,i_{\ell}) \in \mathcal{I}$}
\State $\mathcal{S}\gets \mathcal{S}\cup \prod_{j=1}^{\ell} B(p_{i_j},r\delta)\cap \GGG_{\eps\delta}$\label{algoline:candidates}
\EndFor
\For{\textbf{each} $\sigma \in \SSS$}
\State  $\CCC'\gets \CCC'\cup \{ \seqtocurve{\sigma} \}  $
\EndFor
\State \Return $\CCC'$
\EndProcedure
\end{algorithmic}
\end{algorithm}

\begin{algorithm}
\caption{Query algorithm \label{alg:query1apprx}}
\begin{algorithmic}[1]

\Procedure{\texttt{query}}{curve $Q$ with vertices $q_1,\ldots,q_k$, $\radius>0$, $\eps>0$} 
\State $\mathcal{I}\gets$\texttt{generate\_orders}$(k,k)$
\For{\textbf{each} $(i_1,\ldots,i_{\ell})\in \mathcal{I}$}
\State $flag \gets 1 $
\For{$j=1,\ldots,\ell-1$} 
\If{$\df(\overline{q_{i_j}q_{i_{j+1}}},\seqtocurve{ q_{i_j},\ldots,q_{i_{j+1}}})> \delta$}\Comment{test $\delta$-simplification property}
\State $flag \gets 0$
\EndIf 
\For{\textbf{each} $t=i_j,\ldots , i_{j+1}$}
\Comment{test edge-range-preserving property}
\If{$q_t \notin  \overline{q_{i_j}q_{i_{j+1}}}$}
\State $flag \gets 0$ 
\EndIf
\EndFor
\EndFor
\If{$flag =1$}
\State $Q' \gets  \seqtocurve{{q_{i_1}}, \ldots, {q_{i_{\ell}}}} $ \Comment{a  $\radius$-\straightening of $Q$}
\State $Q'' \gets  \seqtocurve{\left \lfloor \frac{q_{i_1}}{\eps\radius} \right\rfloor\cdot (\eps\radius), \ldots, \left \lfloor \frac{q_{i_{\ell}}}{\eps\radius}  \right\rfloor\cdot (\eps\radius) }$ \Comment{snap $Q'$ to $\GGG_{\eps\radius}$}
\If{$Q''$ in $\HHH$}
\State \Return input curve $P$ associated with $Q''$ in $\HHH$ 
\EndIf
\EndIf
\EndFor
\State \Return “no”
\EndProcedure
\end{algorithmic}
\end{algorithm}

\subsection{Analysis}
\label{subsection:datastructure1apprx_an}
In this section, we analyze the performance of our data structure. 
\begin{lemma}
\label{lemma:numberofcandidates}
For any  curve $P$ with vertices $p_1,\ldots,p_m$, $\radius>0$, $\eps>0$, $r\geq\eps$, $k\in \NN$,  the procedure \texttt{generate\_candidates}$(P,\radius, r,\eps,k)$ has running time in \[  {{m+k-2}\choose{k-2} }\cdot \Oh\left(\frac{  r}{\eps}\right)^k.\] 
\end{lemma}
\begin{proof}
The set $\mathcal{I}$ contains all   
sequences of indices $(i_1,\ldots,i_{\ell}) \in [m]^{\ell}$ such that $\ell\leq k$, and $1=i_1\leq \cdots \leq i_{\ell}=m$. 
Let $\III_{\ell}$ be the subset of $\III$ containing the sequences of  length $\ell$ as denoted in \texttt{generate\_orders}. 
We first claim that   \texttt{generate\_orders}$(m,k)$ runs in time  $\Oh(|\III| \cdot k)$. To see that, consider any sequence of indices $s \in \III$. During the execution of \texttt{generate\_orders}, $s$ is added to the sets of indices (Line~\ref{alg:generate_orders:update}) only once. This step costs $\Oh(k)$, therefore the running time of 
\texttt{generate\_orders}$(m,k)$ is in $\Oh(|\III| \cdot k)$. Now, 
let $\SSS'$ be a multiset which contains all sequences (including duplicates) which are generated and inserted to $\SSS$ in all executions of Line~\ref{algoline:candidates} of  \texttt{generate\_candidates}.  
The running time of \texttt{generate\_candidates}$(P,\radius,r,\eps,k)$ is upper bounded by $\Oh(|\SSS'|\cdot k)$,  because $|\SSS'|\geq |\III| $ and 
 computing $\CCC'$ costs $\Oh(|\SSS'|\cdot k)$ time. We proceed by showing an upper bound on $|\SSS'|$. 

 Any sequence   $(x_1,\ldots,x_{\ell})\in \GGG_{\eps\radius}^{\ell}$, which is included in $\SSS'$, may appear in the computation taking place in  Line~\ref{algoline:candidates} multiple times: once for each sequence  of indices $(i_1,\ldots,i_{\ell})\in \mathcal{I}$ such that for each $j\in [\ell]$, $x_j\in B(p_{i_j},r\radius)$. Notice that $|\mathcal{I}_{\ell}|$ is equal to the number of combinations of $\ell-2$ objects taken (with repetition) from a set of size $m$, i.e. $|\mathcal{I}_{\ell}|={{m+\ell-3}\choose{\ell-2} }$. 
Hence, by the Hockey-stick identity,
\begin{align}
\label{eq:hockeystick}
    | \III|=
   \sum_{\ell=2}^k |\mathcal{I}_{\ell}|=  
   \sum_{\ell=2}^k {{m+\ell-3}\choose{\ell-2} }= 
   \sum_{\ell=0}^{k-2} {{m+\ell-1}\choose{\ell} }=
   {{m+k-2}\choose{k-2}}. 
\end{align}
 Using (\ref{eq:hockeystick}), we can bound $|\SSS'|$ as follows: 
 \begin{align*}
|\SSS'| &\leq \sum_{\ell=2}^k \sum_{(i_1,\ldots i_{\ell}) \in \mathcal{I}_{\ell}}  \left| \prod_{j=1}^{\ell} B(p_{i_j},r\radius)\cap \GGG_{\eps\radius} \right| \\
&\leq \sum_{\ell=2}^k  |\mathcal{I}_{\ell}|\cdot  \Oh\left( \frac{r}{\eps}\right)^{\ell}
\leq  |\mathcal{I}| \cdot \Oh\left(\frac{r}{\eps} \right)^{k} 
\leq{{m+k-2}\choose{k-2} }\cdot \Oh\left(\frac{r}{\eps} \right)^{k}.
 \end{align*}
Hence, the running time is $\Oh(|\SSS'|\cdot k)={{m+k-2}\choose{k-2} }\cdot \Oh\left(\frac{r}{\eps} \right)^{k}$.
\end{proof}

\begin{lemma}
\label{lemma:querycorrectness1}
If \texttt{query}$(Q,\radius, \eps/2)$ returns an input curve $P\in \mathcal{P}$, then $\df(Q,P)\leq (1+\eps)\radius$. If \texttt{query}$(Q,\radius, \eps/2)$ returns “no” then there is no $P\in \mathcal{P}$ such that $\df(Q,P)\leq \radius$.
\end{lemma}
\begin{proof}
When \texttt{query}$(Q, \radius,\eps/2)$ returns an input curve $P\in \mathcal{P}$, it must be that there exists  a  $\radius$-\straightening $Q'$  of $ Q$ such that $ P$ is  associated with $Q''$ in $\HHH$. This implies that $\df( {Q''},  P )\leq (1+\eps/2)\radius$. By the triangle inequality, \[\df( {Q'} ,P)\leq \df( { Q''},{ Q'})+\df( {Q''}, P)\leq (1+\eps)\radius.\] 
Since $ {Q'} $ is a  $\radius$-\straightening of $ Q $, we have that $\df({Q'},  Q )\leq \radius$. Hence, by Lemma~\ref{lemma:simplproxy} applied on $P,Q,{Q'}$ for distance threshold $(1+\eps)\radius$, we obtain $\df(Q,P)\leq (1+\eps)\radius$. 

If \texttt{query}$(Q,\radius, \eps/2)$ returns “no” then there is no  $\radius$-\straightening ${Q'}$ of $Q$ such that $Q''$ is associated with an input curve in $\HHH$. Suppose, for the sake of contradiction, that there exists a curve $P\in \mathcal{P}$ such that $\df(Q,P)\leq \radius$. By Lemma~\ref{lemma:goodsimplificationranges2}, there exists a  $\radius$-\straightening ${{Q}'}$ of $Q$ such that \begin{inparaenum}[i)]
    \item there exists an $11\radius$-visiting order of ${{Q}'}$ on $P$ and
    \item $\df({{Q}'},P)\leq \radius$.
\end{inparaenum}
Let ${{Q}''}$ be the curve obtained by snapping vertices of ${{Q}'}$ to the grid $\GGG_{\eps\radius/2}$. 
By the triangle inequality, there exists a $((11+\eps/2)\radius)$-visiting order of ${{Q}''}$ on $P$ and \[\df({{Q}''},P)\leq \df({Q''},{Q'})+\df({Q'},P)\leq (1+\eps/2)\radius.\] Hence, $Q''\in \CCC(P)$ and $Q''$ is associated with some input curve $P'$ in $\HHH$. This leads to contradiction and we conclude that if \texttt{query}$(Q,\radius, \eps/2)$ returns “no” then  there is no curve $P\in \mathcal{P}$ such that $\df(P,Q)\leq \radius$. 
\end{proof}

\begin{lemma}
\label{lemma:querytime1}
For any query curve $Q$ of complexity $k$, $\delta>0$, $\eps>0$, \texttt{query}$(Q,\delta,\eps)$ runs in time $\Oh(k\cdot 2^{k})$. 
\end{lemma}
\begin{proof}

Let $q_1,\ldots,q_k$ be the vertices of $Q$. We enumerate all sequences starting with $q_1$, followed by any possible subsequence of $q_2,\ldots,q_{k-1}$ and ending with $q_k$. There are at most $2^{k-2}$ such sequences, and for each one of them we test whether it defines a  $\radius$-\straightening of $Q$.  This is done in two steps: we first test if each shortcut is within distance $\delta$ from the corresponding subcurve, and then we decide if the edge-range-preserving property is satisfied. Computing the Fr\'echet distance between a shortcut and the original subcurve costs  linear time in the complexity of the subcurve by Theorem~\ref{theorem:frechetdecision}. Hence, we can decide in $\Oh(k)$ time if the sequence in question defines a $\radius$-simplification of $Q$.  To decide if the edge-range-preserving property is satisfied, we check for each shortcut if the corresponding subcurve stays within range by testing all of its vertices one by one. Therefore, this step also costs $\Oh(k)$ time. Since we employ perfect hashing, each probe to $\HHH$ costs $\Oh(k)$ time. 
We can also check in $\Oh(k)$ time if the answer returned by $\HHH$ is the one we are searching for. 
Hence, the overall query time is in $\Oh(k\cdot 2^k)$. 

\end{proof}

\begin{theorem} \label{thm:onepluseps}
Let $\eps\in(0,1]$. There is a data structure for the $(1+\eps)$-ANN problem, 
which stores $n$ one-dimensional curves of complexity $m$ and supports query curves  of complexity $k$,  uses space in $n\cdot \Oh\left(\frac{m}{k \eps}\right)^{k}$, needs $\Oh(nm)\cdot \Oh\left(\frac{m}{k \eps}\right)^{k}$ expected preprocessing time and answers a query in $\Oh(k\cdot 2^{k})$ time. 
\end{theorem}
\begin{proof}
The data structure is described in Section~\ref{subsection:datastructure1apprx_ds}. Correctness follows from Lemma~\ref{lemma:querycorrectness1}.  The bound on the query time follows from Lemma~\ref{lemma:querytime1}. It remains to analyze the running time of \texttt{preprocess}$(\mathcal{P},\radius,\eps/2,k)$ and the space complexity of the data structure.

By Lemma~\ref{lemma:numberofcandidates}, 
for any $P\in \PPP$, 
the running  time needed to compute $\CCC'$ is upper bounded by   $ {{m+k-2}\choose{k-2} }\cdot \Oh\left(\frac{1}{\eps} \right)^{k}=\Oh\left( \frac{m}{k\eps}\right)^k$. 
Hence, for each $P\in \mathcal{P}$, $|\CCC(P)| = \Oh\left(\frac{m}{k \eps}\right)^{k}$. Therefore, the space required for each input curve $P\in \mathcal{P}$ is upper bounded by $\Oh(|\CCC(P)|\cdot k)$. 
Computing $\CCC(P)$ costs $\Oh(|\CCC'|\cdot mk)=\Oh\left(\frac{m}{k \eps}\right)^{k}\cdot \Oh(m)$ time, because we need to decide for each curve  $Q\in \CCC'$, whether its Fr\'echet distance  from $P$ is at most $(1+\eps/2)\radius$, which can be done in $\Oh(|Q|\cdot |P|)$ time using Theorem~\ref{theorem:frechetdecision}. Assuming perfect hashing for $\HHH$, the overall expected preprocessing time is in $\Oh(nm)\cdot \Oh\left(\frac{m}{k \eps}\right)^{k}$ and the space usage is in $\Oh(n)\cdot \Oh\left(\frac{m}{k \eps}\right)^{k}$. 
\end{proof}



\section{\boldmath $(2+\eps)$-Approximation}\label{section:datastructure2apprx}

In this section we present three $(2+\eps)$-ANN data structures with different tradeoffs between preprocessing and query time.

\subsection{Fast query algorithm}
\label{section:datastructure2apprxfastquery}
In this section, we propose a data structure for the $(2+\eps)$-ANN problem, with query time in $\Oh(k)$. The space complexity and the preprocessing time are the same as in the $(1+\eps)$-ANN data structure  of Theorem~\ref{thm:onepluseps}.

\subsubsection*{Data structure}
We are given as input a set of one-dimensional curves $\mathcal{P}$, as sequences of vertices, the distance threshold $\radius>0$, the approximation error $\eps>0$ and the complexity of the supported queries $k$. The data structure is exactly the same as in Section~\ref{section:datastructure1apprx}. To build it, we call \texttt{preprocess}$(P,\radius, \eps/2,k)$, as defined in Algorithm~\ref{alg:preprocessing1apprx}, in Section~\ref{section:pseudocode1apprx}.  Let $\HHH$ be the resulting dictionary, constructed by \texttt{preprocess}$(P,\radius,\eps/2,k)$. 
\subsubsection*{Query algorithm}
Let $Q$ be the query curve with vertices $q_1,\ldots,q_k$ and let $\eps>0$ be the approximation error. The query algorithm  first computes a $\radius$-signature $Q'$ of $Q$, and then it snaps the vertices of $Q'$ to the grid $\GGG_{\eps\radius/2}$,  to obtain a curve  $Q''$. If $Q''$ is stored in $\HHH$, then we return its associated input curve $P\in\mathcal{P}$, otherwise we return "no". The query algorithm is implemented in  \texttt{query2}, which can be found in Algorithm~\ref{alg:query2apprx}. To achieve approximation factor $2+\eps$, we run \texttt{query2}$(Q,\radius,\eps/2)$. 


\begin{algorithm}[H]
\caption{Query algorithm\label{alg:query2apprx}}
\begin{algorithmic}[1]
\Procedure{\texttt{query2}}{curve $Q$ with vertices $q_1,\ldots,q_k$, $\radius>0$, $\eps>0$}
\State $Q' \gets $ $\radius$-signature of $Q$
\State $q_1',\ldots,q_{\ell}' \gets $ vertices of $Q'$ 
\State $Q'' \gets  \seqtocurve{\left \lfloor \frac{q_{1}'}{\eps\radius} \right\rfloor\cdot (\eps\radius), \ldots, \left \lfloor \frac{q_{{\ell}}'}{\eps\radius}  \right\rfloor\cdot (\eps\radius) }$ \Comment{snap $Q'$ to $\GGG_{\eps\radius}$}
\If{$Q''$ in $\HHH$}
\State \Return input curve $P$ associated with $Q''$ in $\HHH$ 
\EndIf
\State \Return “no”
\EndProcedure
\end{algorithmic}
\end{algorithm}

\begin{lemma}
\label{lemma:querycorrectness2apprxfastquery}
If \texttt{query2}$(Q, \radius, \eps/2)$ returns an input curve $P\in \mathcal{P}$, then $\df(Q,P)\leq (2+\eps)\radius$. If \texttt{query2}$(Q, \radius,\eps/2)$ returns “no” then there is no $P\in \mathcal{P}$ such that $\df(Q,P)\leq \radius$.
\end{lemma}
\begin{proof}
If \texttt{query2}$(Q,\radius, \eps/2)$ returns an input curve $P\in \mathcal{P}$, then it must be that $Q''$ is stored in $\HHH$, and $P$ is its associated input curve. By the construction of $\HHH$, it must be that $\df(P, Q'' )\leq (1+\eps/2)\radius$. By the definition of signatures we know that $\df(Q,{Q'})\leq \radius$, and by the triangle inequality we obtain \[\df( Q,{Q''} )\leq \df(Q, {Q'} )+\df({ Q''} ,{Q'})\leq (1+\eps/2)\radius.\]  
Hence, by the triangle inequality we obtain 
\[
\df(P,Q)\leq \df(P,{Q''} )+\df( Q,{Q''} ) \leq (2+\eps)\radius.
\]

Now suppose that \texttt{query2}$(Q,\radius, \eps/2)$ returns “no”. This means that $Q''$ is not stored in $\HHH$. Suppose that there exists a $P\in\mathcal{P}$ such that $\df(P,Q)\leq \radius$. Then by 
Lemma~\ref{lemma:signatures2} there exists a $\radius$-visiting order of $Q'$ on $P$. Therefore, by the triangle inequality, there exists a $((1+\eps/2)\radius)$-visiting order of $ Q''$ on $P$, which implies that $Q''\in \CCC(P)$, and hence $Q''$ is stored in $\HHH$. This leads to a contradiction, since we have assumed that $Q''$ is not stored in $\HHH$.   Hence, if  \texttt{query2}$(Q, \radius,\eps/2)$ returns “no” then there is no $P\in\mathcal{P}$ such that $\df(P,Q)\leq \radius$. 
\end{proof}

\begin{theorem} \label{thm:twopluseps_one}
Let $\eps\in(0,1]$. There is a data structure for the $(2+\eps)$-ANN problem, 
which stores $n$ one-dimensional curves of complexity $m$ and supports query curves of complexity $k$,  uses space in $n\cdot \Oh\left(\frac{m}{k \eps}\right)^{k}$, needs $\Oh(nm)\cdot \Oh\left(\frac{m}{k \eps}\right)^{k}$ expected preprocessing time and answers a query in $\Oh(k)$ time. 
\end{theorem}
\begin{proof}
Correctness of the data structure follows from Lemma~\ref{lemma:querycorrectness2apprxfastquery}. The space complexity and the preprocessing time are analyzed in the proof of Theorem~\ref{thm:onepluseps}. It remains to show that \texttt{query2}$(Q, \radius,\eps/2)$ runs in $\Oh(k)$ time. 

To compute a $\radius$-signature of $Q$, we use the algorithm of Driemel, Krivosija and Sohler \cite{DKS16}, which runs in $\Oh(k)$ time. 
Since we employ perfect hashing and we assume that the floor function can be computed in constant time, each probe to $\HHH$ costs $\Oh(k)$ time, and  we can also check at the same time if the answer returned by $\HHH$ is the one we are searching for. 
We conclude that \texttt{query2}$(Q,\radius, \eps/2)$ runs in $\Oh(k)$ time. 
\end{proof}

\subsection{Improved preprocessing time}

\label{section:twopluseps_two}
In this section, we show that there exists a data structure for the $(2+\eps)$-ANN problem, with space complexity and preprocessing time in $n\cdot \Oh(1/\eps)^k+\Oh(nm)$. The query time is in $\Oh(k\cdot 2^k)$. This avoids the factor $(m/k)^k$ of our previous data structures.

\subsubsection*{Data structure}
We are given as input a set of one-dimensional curves $\mathcal{P}$, as sequences of vertices, the distance threshold $\radius>0$, the approximation error $\eps>0$, and the complexity of the supported queries $k$. To build the data structure, we use a modified version of the preprocessing algorithm in  Section~\ref{section:datastructure1apprx}. 
For each input curve $P \in \mathcal{P}$, we compute a $\radius$-signature $P'$ of $P$. If the complexity of $P'$ is at most $k+2$ then 
we compute a set $\CCC':=\CCC'(P')$ which contains all curves $Q$ such that: 
\begin{inparaenum}[i)]
\item $Q$ has complexity at most $k$,
\item all vertices of $Q$ belong to $\GGG_{\eps\radius/2}$, and
\item there is a $((16+\eps/4)\radius)$-visiting order of $Q$ on $P'$.
\end{inparaenum} 
This step is similar to the one in the preprocessing algorithm in  Section~\ref{section:datastructure1apprx}, although  here we consider signatures of the input curves  instead of the original curves.

The filtering process is also slightly different. 
 We filter $\CCC'$ to obtain a set $\CCC(P)$ which contains only those curves of $\CCC'$ with: \begin{inparaenum}[i)]
 \item Fr\'echet distance at most $(2+\eps/4)\radius$ from $P$,
 \item their first point within distance $(1+\eps/4)\radius$ from $P(0)$, and
 \item their last point within distance $(1+\eps/4)\radius$ from $P(1)$. 
 \end{inparaenum}  Let $\HHH$ be a dictionary which is initially empty. For each $P\in \PPP$, we store $\CCC(P)$ in $\HHH$ as follows: for each $Q\in \CCC(P)$, 
if $Q$ is not already stored in $\HHH$, then we insert $Q$ into $\HHH$, associated with a pointer to $P$. The preprocessing algorithm is implemented in  \texttt{preprocess2}, which can be found in Algorithm~\ref{alg:preprocessing2apprx}. We also make use of the subroutine \texttt{generate\_candidates} described in Algorithm \ref{alg:preprocessing1apprx}, in Section~\ref{section:pseudocode1apprx}. To achieve approximation factor $(2+\eps)$, we run  \texttt{preprocess2}$(\PPP,\radius,22,2,\eps/4,k)$. 
\subsubsection*{Query algorithm}
Let $Q$ be the query curve with vertices $q_1,\ldots,q_k$ and let $\eps>0$ be the approximation error. The query algorithm is the same as in the data structure of Section~\ref{section:datastructure1apprx}, but we run it with different input parameters. 
In particular, we run \texttt{query}$(Q,2\radius,\eps/4)$ (see Algorithm~\ref{alg:query1apprx})  on the dictionary $\HHH$ which is constructed by \texttt{preprocess2}$(\PPP,\radius,22,2,\eps/4,k)$. 

\begin{algorithm}[h]
\caption{Preprocessing algorithm.  We call \texttt{preprocess2} to build the data structure. \label{alg:preprocessing2apprx}}
\begin{algorithmic}[1]
\Procedure{\texttt{preprocess2}}{input set $\PPP$, $\radius>0$, $r>0$, $t>0$, $\eps>0$, $k\in\NN$}
     \State Initialize empty dictionary $\HHH$ 
     \For {{\bf each} $P \in \PPP$}
     \State $P' \gets $ $\delta$-signature of $P$
     \If{$|P'|\leq k+2$}
     \State $\CCC(P)\gets$ \texttt{generate\_keys2}$(P', \radius, r, t, \eps, k)$
     \For{{\bf each} $Q\in \CCC(P)$}
      \If{$Q$ not in $\HHH$}
     \State insert key $Q$ in $\HHH$, associated with a pointer to $P$  
     \EndIf
     \EndFor
     \EndIf
     \EndFor
     \EndProcedure
\Procedure{\texttt{generate\_keys2}}{curve $P$, $\radius>0$, $r>0$, $t>0$, $\eps>0$, $k$}
\State  $\CCC'\gets$\texttt{generate\_candidates}$(P,\radius,r+\eps,\eps,k)$
\State $\CCC  \gets \emptyset$
\For{{\bf each} $Q\in \CCC'$}
\If{$\df( P ,  {Q} )\leq (t+\eps)\radius$ \textbf{and} $|P(0)-{Q}(0)|\leq (1+\eps)\radius$ \textbf{and} $|P(1)-{Q}(1)|\leq (1+\eps)\radius$ } 
\State $\CCC \gets \CCC \cup \{ Q\}$
\EndIf
\EndFor
\State {\bf  return} $\CCC  $
\EndProcedure
\end{algorithmic}
\end{algorithm}

\begin{lemma}
\label{lemma:querycorrectness2apprxmed}
If \texttt{query}$(Q, 2\radius,\eps/4)$ returns an input curve $P\in \mathcal{P}$, then $\df(Q,P)\leq (2+\eps)\radius$. If \texttt{query}$(Q, 2\radius, \eps/4)$ returns “no” then there is no $P\in \mathcal{P}$ such that $\df(Q,P)\leq \radius$.
\end{lemma}
\begin{proof}

When \texttt{query}$(Q,2\radius, \eps/4)$ returns an input curve $P\in \mathcal{P}$, it must be that there is  a  $\radius$-\straightening $Q'$ of $ Q$ such that $ P$ is  associated with $Q''$ in $\HHH$, where $Q''$ denotes the curve produced by snapping  vertices of $Q'$ to $\GGG_{\eps\radius/4}$. This implies that $Q''\in \CCC(P)$,  and therefore $\df(P',{Q''} )\leq (2+\eps/4)\radius$, $| P'(0)-{Q''}(0)|\leq (1+\eps/4)\radius $, $| P'(1)-{Q''}(1)|\leq (1+\eps/4)\radius $, where $P'$ is the $\radius$-signature of $P$ computed by \texttt{preprocess}. By the triangle inequality,
\[
\df(P',{Q'}  )\leq \df(P',{Q''}  )+\df( {Q'} , {Q''}  ) \leq (2+\eps/2)\radius. 
\]
Similarly, by the triangle inequality, $| P'(0)-{Q'}(0)|\leq (1+\eps/2)\radius $, $| P'(1)-{Q'}(1)|\leq (1+\eps/2)\radius$. 
Lemma~\ref{lemma:signatureproxy} implies that $\df(P,{Q'})\leq (2+\eps)\radius$, because $P'$ is a $\radius$-signature of $P$, $\df(P',Q')\leq (2+\eps/2)\radius $, $|P(0)-Q'(0)|\leq (1+\eps/2)\radius$ and 
$|P(1)-Q'(1)|\leq (1+\eps/2)\radius$.
Then, by Lemma~\ref{lemma:simplproxy}, we conclude that $\df(P,Q)\leq (2+\eps)\radius$. 

 If \texttt{query}$(Q, 2\radius, \eps/4)$ returns “no”, then there is no input curve  $P\in \PPP$ such that $|P'|\leq k+2$, where $P'$ is the $\radius$-signature computed by \texttt{preprocess2} and such that there exists a  $\radius$-\straightening $Q'$ of $Q$ with $Q'\in \CCC(P)$. Suppose for the sake of contradiction that there is an input curve $P\in \PPP$ such that $\df(Q,P)\leq \radius$. Then by the triangle inequality and  the fact that $\df(P,P')\leq \radius$, we obtain   $\df(Q,P')\leq 2\radius$. In addition, by Lemma~\ref{lemma:signatures2} there is a $\delta$-visiting order of $P'$ on $Q$. Since $P'$ satisfies the $\delta$-edge-length property, any two consecutive interior vertices lie at distance at least $2\delta$ to each other. Thus, no two consecutive interior vertices can belong to the same $\delta$-range. Hence, $|P'|\leq |Q|+2\leq k+2$.  
By Lemma~\ref{lemma:goodsimplificationranges2}, there exists a  $2\radius$-\straightening ${Q'}$ of $Q$ which satisfies \begin{compactenum}[i)]
\item there exists a $22\radius$-visiting order of $ {Q'}$ on $P'$,
\item $\df({Q'},P')\leq 2\radius$.
\end{compactenum} 
By the definition of signatures, we have $P(0)=P'(0)$ and $P(1)=P'(1)$, and since $\df(P,Q)\leq \radius$, we have $|P'(0)-Q(0)|\leq \radius$ and $|P'(1)-Q(1)|\leq \radius$. 
 By the definition of \straightenings,  we have ${Q'}(0)=Q(0)$ and ${Q'}(1)=Q(1) $ and therefore $|P'(0)-{Q'}(0)|\leq \radius$ and $|P'(1)-{Q'}(1)|\leq \radius$. 
Hence, by the triangle inequality there exists a $((22+\eps/4)\radius)$-visiting order of $ {Q''}$ on $P'$, $\df({Q''},P')\leq (2+\eps/4)\radius$, $|P'(0)-{Q''}(0)|\leq (1+\eps/4)\radius$ and $|P'(1)-{Q''}(1)|\leq (1+\eps/4)\radius$. This implies that $Q''\in\CCC(P)$ which leads to a contradiction. 
\end{proof}

\begin{theorem} \label{thm:twopluseps_two}
Let $\eps\in(0,1]$. There is a data structure for the $(2+\eps)$-ANN problem, 
which stores $n$ one-dimensional curves of complexity $m$ and supports query curves of complexity $k$,  uses space in $n\cdot \Oh\left(\frac{1}{ \eps}\right)^{k} +\Oh(nm)$, needs $n\cdot \Oh\left(\frac{1}{ \eps}\right)^{k}+\Oh(nm)$ expected preprocessing time and answers a query in $\Oh(k\cdot 2^k)$ time. 
\end{theorem}
\begin{proof}
Correctness follows from Lemma~\ref{lemma:querycorrectness2apprxmed}. The bound on the query time follows from Lemma~\ref{lemma:querytime1}. It remains to bound the space complexity and the preprocessing time of the data structure. 

Computing one $\radius$-signature for each $P \in \PPP$  takes linear time $\Oh(mn)$ in total, using the algorithm of Driemel, Krivosija and Sohler \cite{DKS16}. 
Let $P'$ be the $\radius$-signature of some curve $P\in \PPP$ as computed during preprocessing. If $|P'|>k$ we ignore $P$. 
 By Lemma~\ref{lemma:numberofcandidates}, for any $P'\in \PPP'$, the running  time needed to compute $\CCC'$,  is upper bounded by   $ {{|P'|+k-2}\choose{k-2} }\cdot \Oh\left(\frac{1}{\eps} \right)^{k}=\Oh\left( \frac{1}{\eps}\right)^k$. 
The space required for $P'$ is upper bounded by $\Oh(|\CCC(P')|\cdot k + m)=\Oh(|\CCC'|\cdot k + m)=\Oh\left( \frac{1}{\eps}\right)^k+\Oh(m)$. 
Computing $\CCC(P')$ costs $\Oh(|\CCC'|\cdot k)=\Oh\left(\frac{1}{ \eps}\right)^{k}$ time, since we take a decision on the Fr\'echet distance between each curve in $\CCC'$, and $P'$, by making use of Theorem~\ref{theorem:frechetdecision} . Assuming perfect hashing for $\HHH$, the overall expected preprocessing time is in $\Oh(n)\cdot \Oh\left(\frac{1}{ \eps}\right)^{k}$ and the space usage is in $\Oh(n)\cdot \Oh\left(\frac{1}{ \eps}\right)^{k}$.  
\end{proof}



\subsection{Linear preprocessing time}

In this section we present a data structure for the $(2+\eps)$-ANN problem with linear space and  preprocessing time $\Oh(nm)$ and with query time in $\Oh(1/\eps)^k$. 

\subsubsection*{Data structure}
We are given as input a set of one-dimensional curves $\mathcal{P}$, as sequences of vertices, a distance threshold $\radius>0$, the approximation error $\eps>0$ and the complexity of the supported queries $k$.  For each input curve $P\in \PPP$, we compute a $\radius$-signature $P'$ of $P$. If $|P'|>k+2$ then we ignore $P$, otherwise  we snap it to $\GGG_{\eps\radius/2}$ to obtain a curve $P''$.  Let $\HHH$ be a dictionary which is initially empty. For each $P\in \PPP$, we store $P''$ in $\HHH$ as follows: 
if $P''$ is not already stored in $\HHH$, then we insert $P''$ into $\HHH$, associated with a pointer to $P$. To achieve  approximation factor  $2+\eps$, we run \texttt{preprocess3}$(\PPP,\radius,\eps/2,k)$, as defined in Algorithm~\ref{alg:preprocessing2apprx3}.

\subsubsection*{Query algorithm}
Let $Q$ be a query curve of complexity $k$. 
We compute a set $\CCC':=\CCC'(Q)$ which contains all curves $P$ such that: 
\begin{inparaenum}[i)]
\item $P$ has complexity at most $k$,
\item all vertices of $P$ belong to $\GGG_{\eps\radius/2}$, and
\item there is a $((1+\eps/2)\radius)$-visiting order of $P$ on $Q$.
\end{inparaenum} 
 We filter $\CCC'$ to obtain a set $\CCC(Q)$ which contains only those curves of $\CCC'$ with: \begin{inparaenum}[i)]
 \item Fr\'echet distance at most $(2+\eps/2)\radius$ from $Q$,
 \item their first point within distance $(1+\eps/2)\radius$ from $Q(0)$, and
 \item their last point within distance $(1+\eps/2)\radius$ from $Q(1)$. 
 \end{inparaenum} 
We probe $\HHH$ for each key $P\in\CCC(Q) $: if we find a $P\in \CCC(Q)$ stored in $\HHH$ then we return the associated input curve. If there is no $P\in \CCC(Q)$ stored in $\HHH$ then we return “no”. To achieve the desired approximation, we run \texttt{query3}$(Q,\radius,\eps/2)$, as defined in Algorithm~\ref{alg:query2apprx3}.  

\begin{algorithm}
\caption{Preprocessing algorithm \label{alg:preprocessing2apprx3}}
\begin{algorithmic}[1]
\Procedure{\texttt{preprocess3}}{input set $\PPP$, $\radius>0$, $\eps>0$, $k$}
     \State Initialize empty dictionary $\HHH$ 
     \For{\textbf{each} $P \in \PPP$}
     \State $P' \gets $ $\radius$-signature of $P$
     \If{$|P'|\leq k+2$}
     \State $p_1,\ldots ,p_{\ell} \gets $ vertices of $P'$
     \State $P''\gets \seqtocurve{\left\lfloor \frac{p_1}{\eps\radius}\right\rfloor \cdot  (\eps\radius) ,\ldots ,\left\lfloor \frac{p_{\ell}}{\eps\radius}\right\rfloor \cdot (\eps\radius)} $
     \If{$P''$ not in $\HHH$}
     \State insert key $P''$ in $\HHH$, associated with a pointer to $P$  
     \EndIf
     \EndIf
     \EndFor 
     \EndProcedure
     \end{algorithmic}
\end{algorithm}     

\begin{algorithm}
\caption{Query algorithm \label{alg:query2apprx3}}
\begin{algorithmic}[1]
\Procedure{\texttt{query3}}{curve $Q$ with vertices $q_1,\ldots,q_k$, $\radius>0$, $\eps>0$}
     \State $\CCC(Q)\gets $ \texttt{generate\_keys2}$(Q, \radius, 1,2, \eps, k+2)$
     \For{\textbf{each} $P'' \in \CCC(Q)$}
     \If{$P''$ in $\HHH$}
\State \Return input curve $P$ associated with $P''$ in $\HHH$ 
\EndIf
     \EndFor
     \State \Return “no”
     \EndProcedure
     \end{algorithmic}
\end{algorithm}  

\begin{lemma}
\label{lemma:querycorrectness2apprx3}
If \texttt{query3}$(Q,\radius, \eps/2)$ returns an input curve $P\in \mathcal{P}$, then $\df(Q,P)\leq (2+\eps)\radius$. If \texttt{query3}$(Q,\radius,  \eps/2)$ returns “no” then there is no $P\in \mathcal{P}$ such that $\df(Q,P)\leq \radius$.
\end{lemma}
\begin{proof}If \texttt{query3}$(Q,\radius, \eps/2)$ returns an input curve, then it must be that there is a curve $P''\in \CCC(Q)$ which is stored in $\HHH$, associated with a pointer to $P$. Since $P''$ is stored in $\HHH$, there is a curve $P\in \PPP$ with a $\radius$-signature $P'$ such that $\df(P',P'')\leq \eps\radius/2$. Moreover, since $P''\in \CCC(Q)$, we have that $\df(Q,P'')\leq (2+\eps/2)\radius$, $|Q(0)-P''(0)|\leq (1+\eps/2)\radius$, $|Q(1)-P''(1)|\leq (1+\eps/2)\radius$. By the triangle inequality we obtain, 
 $\df(Q,P')\leq (2+\eps)\radius$, $|Q(0)-P'(0)|\leq (1+\eps)\radius$, $|Q(1)-P'(1)|\leq (1+\eps)\radius$. By Lemma~\ref{lemma:signatureproxy}, 
 since $\df(Q,P')\leq (2+\eps)\radius$, $|Q(0)-P'(0)|\leq (1+\eps)\radius$, $|Q(1)-P'(1)|\leq (1+\eps)\radius$ and $P'$ is a $\radius$-signature of $P$, we conclude $\df(Q,P)\leq (2+\eps)\radius$.

 If \texttt{query3}$(Q, \radius, \eps/2)$ returns “no”, then it must be that there is no curve $P''\in \CCC(Q)$ which is stored in $\HHH$. Suppose for the sake of contradiction that there is an input curve $P\in \PPP$ such that $\df(Q,P)\leq \radius$. Let $P'$ be the $\radius$-signature of $P$, as computed during preprocessing.  By Lemma~\ref{lemma:signatures2}, there is a  $\radius$-visiting order of $P'$ on $Q$ and therefore $|P'|\leq k+2$. Let $P''$ be the curve produced by snapping the vertices of $P'$ to the grid $\GGG_{\eps\radius/2}$. By the triangle inequality there is a  $((1+\eps/2)\radius)$-visiting order of $P''$ on $Q$. Therefore, $P''$ must be included in $\CCC(Q)$, which leads to contradiction. 
\end{proof}

\begin{theorem} \label{thm:twopluseps_three}
Let $\eps\in(0,1]$. There is a data structure for the $(2+\eps)$-ANN problem, 
which stores $n$ one-dimensional curves of complexity $m$ and supports query curves of complexity $k$,  uses space in $\Oh(nm)$, needs $\Oh(nm)$ expected preprocessing time and answers a query in $\Oh(1/\eps)^{k+2}$ time. 
\end{theorem}
\begin{proof}
Correctness follows from Lemma~\ref{lemma:querycorrectness2apprx3}. It remains to bound the space complexity, the  preprocessing time and the query time.

Using the algorithm of Driemel, Krivosija and Sohler \cite{DKS16}, we can compute a signature in linear time. Since we assume that the floor function can be computed in $\Oh(1)$, and that $\HHH$ is implemented using perfect hashing, \texttt{preprocess3}$(\PPP,1,\eps/2,k)$ has running time $\Oh(nm)$. Therefore, the space usage is also in $\Oh(nm)$. 

To bound the query time, we bound the running time of 
\texttt{generate\_keys2}$(Q,\radius, 1,2,\eps/2,k+2)$, because the last part of \texttt{query3} is an enumeration over all curves returned by \texttt{generate\_keys2} and probing $\HHH$ for each one of them.  To bound the running time of 
\texttt{generate\_keys2}$(Q,\radius, 1,2,\eps/2,k+2)$, it suffices to bound the running time of 
\texttt{generate\_candidates}$(Q,\radius, (1+\eps/2),\eps/2,k+2)$. By Lemma~\ref{lemma:numberofcandidates}, this running time is upper bounded by ${{2k}\choose{k-2} }\cdot \Oh\left(\frac{  1}{\eps}\right)^{k+2} =  \Oh\left(\frac{  1}{\eps}\right)^{k+2}$. Recall that we employ perfect hashing and we assume that the floor function can be computed in constant time. Hence each probe to $\HHH$ costs $\Oh(k)$ time, and  we can also check in $\Oh(k)$ if  $\HHH$ returns the correct answer. 
We conclude that \texttt{query2}$(Q,\radius, \eps/2)$ runs in time $ \Oh\left(\frac{  1}{\eps}\right)^{k+2}$. 
\end{proof}


\section{\boldmath $(3+\eps)$-Approximation}\label{section:datastructure3apprx}

In this section, we present a data structure for the $(3+\eps)$-ANN problem with preprocessing time and space complexity in $n\cdot \Oh(1/\eps)^k+\Oh(nm)$ and query time in $\Oh(k)$. 
\subsubsection*{Data structure} We are given as input a set of one-dimensional curves  $\mathcal{P}$, as sequences of vertices, a distance threshold $\radius>0$, the approximation error $\eps>0$ and the complexity of the supported queries $k$.  To build the data structure, we use the preprocessing algorithm of the data structure in Section~\ref{section:twopluseps_two}.
Let $\HHH$ be the dictionary, constructed by \texttt{preprocess2}$(\PPP, \radius ,2,3,\eps/2,k)$. 

\subsubsection*{Query algorithm}Let $Q$ be a query curve. We run the query algorithm of the data structure in Section~\ref{section:datastructure2apprxfastquery}. In particular, we run \texttt{query2}$(Q,2\radius,\eps/2)$ on $\HHH$. 
\begin{lemma}
\label{lemma:querycorrectness3apprx}
If \texttt{query2}$(Q,2\radius, \eps/2)$ returns an input curve $P\in \mathcal{P}$, then $\df(Q,P)\leq (3+\eps)\radius$. If \texttt{query2}$(Q, 2\radius, \eps/2)$ returns “no” then there is no $P\in \PPP$ such that $\df(Q,P)\leq \radius$.
\end{lemma}
\begin{proof}
Let $Q'$ be the $2\radius$-signature of $Q$ and let $Q''$ be the curve obtained by snapping vertices of $Q'$ to $\GGG_{\eps\radius/2}$, as computed in \texttt{query2}. 

If \texttt{query2}$(Q,2\radius, \eps/2)$ returns an input curve $P\in \mathcal{P}$, then it must be that $Q''\in \CCC(P)$, where $\CCC(P)$ is the result of \texttt{generate\_keys2}$(P',\radius, 2, 3, \eps/2, k)$ and $P'$ is a $\radius$-signature of $P$, as computed by \texttt{preprocess2}. 
By the construction of $\CCC(P)$, it must be that $\df(P',Q'')\leq (3+\eps/2)\radius$, $|P'(0)-Q''(0)|\leq (1+\eps/2)\radius$ and $|P'(1)-Q''(1)|\leq (1+\eps/2)\radius$. Hence, by the triangle inequality $\df(Q',P')\leq (3+\eps)\radius$, $|P'(0)-Q'(0)|\leq (1+\eps)\radius$ and $|P'(1)-Q'(1)|\leq (1+\eps)\radius $. 
We now apply Lemma~\ref{lemma:signatureproxy} twice. We first apply  it on $P'$, $Q'$, $Q$.  
Since $\df(P',Q')\leq (3+\eps)\radius$, $|P'(0)-Q'(0)|\leq (1+\eps)\radius$, $|P'(1)-Q'(1)|\leq (1+\eps)\radius $ and $Q'$ is a $2\radius$-signature of $Q$, 
we obtain $\df(P',Q)\leq (3+\eps)\radius$. Then, we apply it on $P'$, $P$, $Q$. 
Since $\df(P',Q)\leq (3+\eps)\radius$, $|P'(0)-Q(0)|=|P'(0)-Q'(0)|\leq (1+\eps)\radius \leq (2+\eps)\radius$, $|P'(1)-Q(1)|=|P'(1)-Q'(1)|\leq (1+\eps)\radius \leq (2+\eps)\radius$, and $P'$ is a $\radius $-signature of $P$, 
we obtain $\df(P,Q)\leq (3+\eps)\radius$.

If \texttt{query2}$(Q,2\radius, \eps/2)$ returns “no” then  $Q''$ is  not stored in $\HHH$ as a key. For the sake of contradiction, we assume that there exists an input curve $P\in \PPP$ such that $\df(P,Q)\leq \radius$. Then by definition, $|P'(0)-Q'(0)|\leq \radius$ and $|P'(1)-Q'(1)|\leq \radius$. In addition, by Lemma~\ref{lemma:signatures3}, $\df(P',Q')\leq 3\radius$ and there is a $2\radius$-visiting order of $Q'$ on $P'$,  By the triangle inequality we obtain $\df(P',Q'')\leq (3+\eps/2)\radius$, $|P'(0)-Q''(0)|\leq (1+\eps/2)\radius$,  $|P'(1)-Q''(1)|\leq (1+\eps/2)\radius$, and that there is a $((2+\eps/2)\radius)$-visiting order of $Q''$ on $P'$. Hence, by the construction of $\CCC(P)$, it must be that $Q''\in \CCC(P)$ which implies that $Q''$ is stored as a key in $\HHH$. This is a contradiction.
\end{proof}

\begin{theorem}\label{thm:threepluseps}
Let $\eps\in(0,1]$. There is a data structure for the $(3+\eps)$-ANN problem, 
which stores $n$ one-dimensional curves of complexity $m$ and supports query curves of complexity $k$,  uses space in $n\cdot \Oh(1/\eps)^k +\Oh(nm)$, needs $n\cdot \Oh(1/\eps)^k+\Oh(nm)$ expected preprocessing time and answers a query in $\Oh(k)$ time. 
where $k$ is the complexity of the query curve. 
\end{theorem}
\begin{proof}
Correctness follows from Lemma~\ref{lemma:querycorrectness3apprx}. The bounds on the preprocessing time and space complexity follow from Theorem~\ref{thm:twopluseps_two}. The bound on the query time follows from Theorem~\ref{thm:twopluseps_one}. \end{proof}

\section{Proofs of main lemmas}
\label{section:missingproofs}

In this section we give full proofs of the lemmas stated in Section~\ref{section:lemmas}.
We start by proving a fundamental observation and lemma on the Fr\'echet distance of approximately monotone one-dimensional curves. 

\begin{observation}\label{lemma:montonecurvesimple}
 Let $Q$ be a directed line segment and let  $ P:~[0, 1] \mapsto \RR$ be a curve. 
It holds that  $\df(P,Q) \leq \delta$ if and only if the following conditions are satisfied:
\begin{compactenum}[(i)] 
    \item $P$ is $2\delta$-monotone with respect to $Q$, and
    \item $|P(0)-Q(0)|\leq \delta$, $|P(1)-Q(1)|\leq \delta$, and
    \item  $P \subseteq B(Q,\delta)$.
\end{compactenum} 
\end{observation}

\begin{proof}
We assume that $Q(0) \leq Q(1)$ as the other case is symmetric.
Now, assume first that $\df(P,Q) \leq \delta$, then (ii) holds because start and end points are matched in any traversal and (iii) holds as the Hausdorff distance is a lower bound for the Fréchet distance. Finally, (i) holds as otherwise there exist two indices $s,t \in [0,1]$ with $s < t$ and $P(t) < P(s) - 2\delta$. As $Q$ is increasing, no traversal can match $P(s)$ and $P(t)$ in distance at most $\delta$.

Second, assume that (i), (ii), and (iii) hold. 
Then $\df(P,Q) \leq \delta$ is implied by Lemma \ref{lemma:montonecurves}, below, but to provide some intuition we give a simpler proof here.
The following traversal with position $s$ on $P$ and position $t$ on $Q$ stays within distance $\delta$. We start in $P(0), Q(0)$, then we continue on $P$ until $P(s) = Q(0)+\delta$. Then we always choose $t$ such that $Q(t) = \min_{s' \geq s} P(s') + \delta$ while traversing $P$, i.e., continuously increasing $s$. When we reach the end of $Q$, we can traverse $P$ until the end while staying in $Q(1)$. It is easy to check that properties (i), (ii), and (iii) ensure distance $\delta$ during the described traversal.
\end{proof}


The following lemma statement is similar to the above observation with the important difference that the line segment $Q$ is replaced by a $2\delta$-monotone curve. The proof works by constructing a traversal greedily and showing correctness of the greedy algorithm. 

\begin{lemma}\label{lemma:montonecurves}
 Let $P$ and $Q$ be $2\delta$-monotone curves with 
\begin{compactenum}[(i)] 
    \item\label{mc_c4} $P$ is $2\delta$-monotone with respect to $\overline{Q(0)Q(1)}$, and
    \item\label{mc_c1} $|P(0)-Q(0)|\leq \delta$, $|P(1)-Q(1)|\leq \delta$, and
    \item\label{mc_c3} $P \subseteq B(Q,\delta)$, and 
    \item\label{mc_c2} $Q \subseteq \overline{Q(0)Q(1)}$.
\end{compactenum} 
It holds that $\df(P,Q) \leq \delta$.
\end{lemma}

\begin{proof}
We assume that $Q(0) \leq Q(1)$ as the other case is symmetric.
If $Q$ is not $2\delta$-monotone increasing, then it also cannot be $2\delta$-monotone decreasing:
if there are two points $s,t \in [0,1]$ with $s < t$ such that $Q(t) < Q(s) - 2\delta$, then, as $Q(t) \geq Q(0)$ by condition (\ref{mc_c2}), we have that $Q(s) > Q(t) + 2\delta \geq Q(0) + 2\delta$ and thus $Q$ is not $2\delta$-monotone decreasing.
However, as $Q$ is $2\delta$-monotone, it has to be $2\delta$-monotone \emph{increasing}.
Due to condition~(\ref{mc_c4}), $P$ is also $2\delta$-monotone \emph{increasing}.
We give a traversal of $P, Q$ with distance at most $\delta$ --- denoting the position during the traversal with $(s,t) \in [0,1]^2$ --- that tries to maintain two invariants:
\begin{enumerate}[(1)]
	\item\label{invariant1} $P$ and $Q$ are in a position $(s,t) \in [0,1]^2$ such that $P(s) = Q(t) + \delta$.
	\item\label{invariant2} The suffix of $Q$ is strictly greater than the current value $Q(t)$, i.e., $\forall t' > t: Q(t') > Q(t)$.
\end{enumerate}
In general, both invariants may be violated at the very beginning of the traversal, that is, for $s=t=0$. 
Let us first describe how we traverse from the beginning of $P, Q$ to a position $(s,t) \in [0,1]^2$ such that these invariants are fulfilled. 
We first traverse $P$ until it first reaches $Q(0) + \delta$, while in $Q$ we stay in $Q(0)$. Note that by condition~(\ref{mc_c1}), we cannot have $P(0) > Q(0) + \delta$. Furthermore, this traversal is feasible as the traversed prefix of $P$ is in the range $[Q(0) - \delta, Q(0) + \delta]$, by condition~(\ref{mc_c3}), and thus within distance $\delta$ to $Q(0)$. If we reach $P(1)$ before reaching $Q(0) + \delta$, then we know that $Q \subseteq [P(1) - \delta, P(1) + \delta]$ and we can thus traverse complete $Q$ and $\df(P,Q) \leq \delta$. If we did not reach $P(1)$, we now traverse $Q$ until its last point with value $Q(0)$, which is possible as the traversed prefix of $Q$ lies in $[Q(0), Q(0) + 2\delta]$, due to condition~(\ref{mc_c2}) and as $Q$ is $2\delta$-monotone increasing, and the position on $P$ is currently $Q(0) + \delta$.

From now on, we traverse $P$ and $Q$ with the same speed in image space, unless one of the two invariants would be violated by continuing the traversal.
If both invariants would be violated at the same time, we break ties by restoring Invariant~(\ref{invariant1}) before Invariant~(\ref{invariant2}).
Now, let $s$ be the position on $P$ and $t$ the position on $Q$ when an invariant would be violated. When Invariant~(\ref{invariant1}) would be violated, we continue traversing $P$ while staying in $Q(t)$ on $Q$ until the next time we reach a position $s'$ on $P$ with value $P(s') = P(s)$. Note that we might not reach such a position $s'$ because we reached the end of $P$. However, if we did not reach the end of $P$, the invariant is restored. This traversal keeps the two positions at distance $\delta$ as $P(s) = Q(t) + \delta$ and as $P$ is 2\deltamonotone increasing.
In case Invariant~(\ref{invariant2}) would be violated, we continue traversing $Q$ until we reach the largest position $t' > t$ such that $Q(t') = Q(t)$. Note that afterwards, both invariants hold (as we restore Invariant~(\ref{invariant1}) before Invariant~(\ref{invariant2})), and, in particular, we cannot reach the end of $Q$ due to the existence of $Q(t')$ which we reach at the end of restoring Invariant~(\ref{invariant2}). This traversal also keeps the two positions at distance $\delta$ as initially $Q(t) = P(s) - \delta$ and $Q$ is 2\deltamonotone increasing and there is no position $t''$ on $Q$ with $Q(t'') < Q(t)$, i.e., all the points before reaching position $t'$ on $Q$ have to be in the range $[Q(t), Q(t) + 2\delta]$.

In all of the above cases we are guaranteed to make progress in our traversal. Furthermore, we will reach the end of $P$ before or at the same time as we reach the end of $Q$ because, first, while restoring invariants we can only reach the end of $P$ but not of $Q$ as argued above and, second, if we reach the end of $Q$ while both invariants would continue to hold, we also have to reach the end of $P$ at the same time as otherwise we would violate condition~(\ref{mc_c3}) of the lemma.
When we reach the end of $P$, we know that $P(1) \in [Q(1)-\delta, Q(1)+\delta]$ due to condition~(\ref{mc_c1}), and the remaining $Q$ is in $[P(1)-\delta, Q(1)]$. Thus, the remaining $Q$ is in $[P(1)-\delta, P(1)+\delta]$ and consequently $Q$ can be traversed until the end.

It follows from the traversal constructed thereby that $\df(P,Q) \leq \delta$.
\end{proof}

\subsection{Proofs of lemmas for \straightenings}\label{sec:lemmatasim}

Next, we want to prove Lemma~\ref{lemmastraightenings} from Section~\ref{section:lemmas}.
We first prove a simpler statement, which can be thought of as a special case where the \straightening consists of only one edge. 

\begin{lemma}
\label{lemma:strongertriangleineq}
 Let $X=\overline{ab}\subset \RR$ be a line segment and let $Q:~[0,1]\mapsto \RR$ be a  curve such that:
$Q(0)=X(0)$, $Q(1)=X(1)$, for all $t \in [0,1]: Q(t) \in \overline{ab}$ and $\df(Q,X) \leq \delta$. 
For any curve $P:~[0,1]\mapsto \RR$ with $\df(P,X) \leq \delta$, it holds that $\df(P,Q) \leq \delta$.
\end{lemma}

\begin{proof} 

To show the lemma statement, we want to apply Lemma~\ref{lemma:montonecurves} to $P$ and $Q$. For this, we need to show that the conditions on $Q$ and $P$ from the lemma statement are met. By Observation~\ref{lemma:montonecurvesimple} applied to $Q$ and the line segment $X$, it follows that $Q$ must be $2\delta$-monotone with respect to $X$, and by our assumptions, $Q$ is range-preserving (condition (\ref{mc_c2})). By Observation~\ref{lemma:montonecurvesimple} applied to $P$ and $X$, it also follows that $P$ is $2\delta$-monotone, and conditions (\ref{mc_c1}), (\ref{mc_c3}) and (\ref{mc_c4}) are satisfied.  Therefore, Lemma~\ref{lemma:montonecurves} can be applied to $P$ and $Q$ and the claim is implied.
\end{proof}


\lemmastraightenings*

\begin{proof}
Let $q_1,\ldots,q_{\ell}$ be the parameters corresponding to the vertices of $Q'$ in $Q$, i.e., the vertices of $Q'$ are $Q(q_1),\ldots,Q(q_{\ell})$. 
Let $\phi:[0,1] \rightarrow [0,1]^2$ be a $\delta$-traversal between $P$ and $Q'$. Let $0=t_1 \leq \dots \leq t_{\ell}=1$ be a partition of the parameter space of $P$ such that for any $1 \leq i \leq \ell-1$, the edge $\overline{Q(q_i)Q(q_{i+1})}$ is mapped to $P[t_i,t_{i+1}]$ under $\phi$. As such, we have  
\[ \df (P[t_i,t_{i+1}],\overline{Q(q_i)Q(q_{i+1})}) \leq \delta  \]
By the locality property of $\delta$-simplifications, we also have that
\[ \df (Q[q_i,q_{i+1}],\overline{Q(q_i)Q(q_{i+1})}) \leq \delta \]
Now, Lemma \ref{lemma:strongertriangleineq} implies that
\[
\df(P[t_i,t_{i+1}], Q[q_i,q_{i+1}]) \leq \delta.
\]
Finally, we apply Observation \ref{observation:concatenation} on  $P=\bigcirc_{i=1}^{\ell}P[t_i,t_{i+1}]$ and $Q=\bigcirc_{i=1}^{\ell}Q[q_i,q_{i+1}]$, and we obtain  \[\df(P,Q)\leq \max_{i \in [\ell]}\df\left(P[t_i,t_{i+1}], Q[q_i,q_{i+1}] \right) \leq  \delta.\]
\end{proof}

\subsection{Proofs of lemmas for signatures}\label{sec:lemmatasig}
 
Next, we want to prove Lemma~\ref{lemmasignatureproxy} from Section~\ref{section:lemmas}. 
We first prove an auxiliary statement for signature edges in Lemma~\ref{lemma:signatureedge1}. In particular, we need to take care of the first and last edge of the signature. For the other edges we can use Lemma~\ref{lemma:strongertriangleineq}. Technically, we will also need the symmetric statement of this lemma for $a > b$; this follows by mirroring at the origin.
The proof of this lemma turns out be technically involved. For the proof of Lemma~\ref{lemmasignatureproxy} we can then use the same approach as for Lemma~\ref{lemmastraightenings} above.

\begin{lemma}\label{lemma:signatureedge1}
 Let $\delta=\delta'+\delta''$ for $\delta,\delta',\delta''\geq 0$. Let $X=\overline{ab}\subset \RR$ be a line segment with $a \leq b$ and let $Q:~[0,1]\mapsto \RR$ be a curve such that:
$Q(0)=X(0)$, $Q(1)=X(1)$ and $\df(Q,X) \leq \delta'$. Let $P:~[0,1]\mapsto \RR$ be a curve with $\df(P,X) \leq \delta$. 

If either
\begin{compactenum}[(i)]
\item $|Q(0)-P(0)| \leq \delta''$ and $|Q(1)-P(1)| \leq \delta''$, or
\item $|Q(0)-P(0)| \leq \delta''$ and $\max_{t \in [0,1]}(Q(t)) \leq Q(1)$, or
\item $\min_{t \in [0,1]}(Q(t)) \geq Q(0)$ and  $|Q(1)-P(1)| \leq \delta''$,
\end{compactenum}
then it holds that $\df(P,Q) \leq \delta$.
\end{lemma}

\begin{proof}
Let $t_{\min} = \argmin \{ Q(t) \}$ and   $t_{\max} = \argmax \{ Q(t) \}$. In case the minimum (resp.\  maximum) is not unique, we choose any of them.
By Observation~\ref{lemma:montonecurvesimple}, we have that $\forall t\in [0,1]~ Q(t)\in [Q(0)-\delta',Q(1)+\delta']$ and by assumption of case (i) $|P(0)-Q(0)| \leq \delta''$ and $|P(1)-Q(1)| \leq \delta''$. Therefore, by triangle inequality, we have in case (i), that
\[ |P(0)-Q(t_{\min})| \leq \delta \quad\text{   and   }\quad |P(1)-Q(t_{\max})| \leq \delta \]
It is easy to see that this holds in the cases (ii) and (iii), as well, since $|P(0)-Q(0)| \leq \delta$ and $|P(1)-Q(1)| \leq \delta$ holds in any case as we assume $\df(P,X) \leq \delta$.

Now, define
\[ t_1 = \min \{t \in [0,1] \mid  P(t) \geq Q(t_{\min}) + \delta\} \]
\[ t_2 = \max \{t \in [0,1] \mid  P(t) \leq Q(t_{\max}) - \delta\} \]
If such a $t_1$ does not exist, then we set $t_1=1$. If $t_2$ does not exist, then we set $t_2=0$. 

Note that by construction and Observation~\ref{lemma:montonecurvesimple} we have 
\begin{equation}\label{eq:signatureedge1:eq1}
\df(P[0,t_1],Q(0)) \leq \delta \quad\text{   and   }\quad  \df(P[t_2,1],Q(1)) \leq \delta
\end{equation}
Indeed, (\ref{eq:signatureedge1:eq1}) holds true  since 
$Q(t_{\min}) \leq Q(0) \leq Q(t_{\min})+\delta'$
and, likewise, 
$Q(t_{\max}) \geq Q(1) \geq Q(t_{\max})-\delta'$, and, moreover,  the image of the subcurve 
$P[0,t_1]$ is contained in  the interval  $[Q(0)-\delta,Q(t_{\min})+\delta]$ and the image of the subcurve $P[t_2,1]$ is contained in  the interval $[Q(t_{\max})-\delta, Q(1)+\delta]$.

In addition, we have  
\begin{equation}\label{eq:signatureedge1:eq2} \df(P(t_1),Q[0,t_{\min}]) \leq \delta \quad\text{   and   }\quad  \df(P(t_2),Q[t_{\max},1]) \leq \delta
\end{equation}
Indeed, (\ref{eq:signatureedge1:eq2}) holds true, since  $\delta'\leq \delta$ and by Observation~\ref{lemma:montonecurvesimple},  $Q$ is $2\delta'$-monotone increasing, and therefore the image of the subcurve $Q[0,t_{\min}]$ is contained in the interval $[Q(t_{\min}), Q(t_{\min})+2\delta']$ which by construction is equal to $[P(t_1)-\delta', P(t_1)+\delta']$ and the image of the subcurve $Q[t_{\max}, 1]$ is contained in the interval $[Q(t_{\max})-2\delta', Q(t_{\max})]$, which by construction is equal to $[P(t_2)-\delta', P(t_2)+\delta']$.

Now, assume that $t_1 \leq t_2$ and $t_{\min} \leq t_{\max}$. In this case, the subcurves $P[t_1,t_2]$ and $Q[t_{\min},t_{\max}]$ are well-defined. 
By construction, $|P(t_1)-Q(t_{\min}) \mid \leq \delta $, $|P(t_2)-Q(t_{\max}) \mid \leq \delta $ and $Q[t_{\min},t_{\max}]\subseteq \overline{Q(t_{\min}) Q(t_{\max})}$. By  Observation~\ref{lemma:montonecurvesimple}, $P$ and $Q$ are both $2\delta$-monotone with respect to $X$, and, by definition, $X=\overline{Q(0)Q(1)}$. Moreover, by the definition of $t_1,t_2$, we have $P[t_1,t_2]\subseteq B(Q[t_{\min},t_{\max}],\delta)$.
Therefore all conditions of Lemma~\ref{lemma:montonecurves} are satisfied, which implies that
\begin{equation}\label{eq:signatureedge1:eq3} \df({P[t_1,t_2],Q[t_{\min}, t_{\max}]}) \leq \delta 
\end{equation}

In summary, we have by (\ref{eq:signatureedge1:eq1}),(\ref{eq:signatureedge1:eq2}), and (\ref{eq:signatureedge1:eq3})  that
\[ \max \begin{pmatrix} 
\df(P[0,t_1],Q(0))\\
\df(P(t_1),Q[0,t_{\min}])\\
\df({P[t_1,t_2],Q[t_{\min}, t_{\max}]})\\
\df(P(t_2),Q[t_{\max},1])\\
\df(P[t_2, 1],Q(1)) 
\end{pmatrix} \leq \delta \]
Now, by Observation~\ref{observation:concatenation} we can concatenate these subcurves and $\df(P,Q) \leq \delta$ is implied.

If the assumption $t_1 \leq t_2$ fails, then, in fact, a simpler decomposition works.
Indeed, if  $t_1 > t_2$, then it holds by (\ref{eq:signatureedge1:eq1}) and (\ref{eq:signatureedge1:eq2}) that
\[ \max \begin{pmatrix} 
\df(P[0,t_1],Q(0))\\
\df(P(t_1),Q)\\
\df(P[t_1, 1],Q(1)) 
\end{pmatrix} \leq \delta \]
Therefore, also in this case,  $\df(P,Q) \leq \delta$  holds true.

Finally, we need to consider the case that the assumption 
$t_{\min} \leq t_{\max}$ fails. We may assume that $t_1 \leq t_2$, as we covered the case $t_1 > t_2$ above. We will consider the different cases from the lemma statement separately.
First, note that if $t_{\min} > t_{\max}$, then $|Q(t_{\max})-Q(t_{\min})| \leq 2\delta$, since $Q$ is $2\delta$-monotone, and therefore, $Q$ is contained in the interval $[P(t_1)-\delta, P(t_1)+\delta]$. By a similar argument, $Q$ is contained in the interval $[P(t_2)-\delta, P(t_2)+\delta]$.

Now, assume case (ii) from the lemma statement. In this case, we have by the above and by Lemma~\ref{lemma:montonecurves}
\[ \max \begin{pmatrix} 
\df(P[0,t_1],Q(0))\\
\df(P(t_1),Q[0,t_{\min}])\\
\df({P[t_1,1],Q[t_{\min}, 1]})) 
\end{pmatrix} \leq \delta \]
Assume case (iii) from the lemma statement. In this case, we have symmetrically
\[ \max \begin{pmatrix} 
\df({P[0,t_2],Q[0, t_{\max}]})\\
\df(P(t_2),Q[t_{\max},1])\\
\df(P[t_2, 1],Q(1)) 
\end{pmatrix} \leq \delta \]

Now, for case (i), we claim that there exist $0 \leq q_1 \leq q_2 \leq 1$, such that
\[ \max \begin{pmatrix} 
\df(P[0,t_1],Q(0))\\
\df(P(t_1),Q[0,q_1])\\
\df(P[t_1,t_2],Q[q_1,q_2])\\
\df(P(t_2),Q[q_2,1])\\
\df(P[t_2, 1],Q(1)) 
\end{pmatrix} \leq \delta \]
Indeed, from what we derived, $\df(P(t_1),Q[0,q_1])\leq \delta$ and $\df(P(t_2),Q[q_2,1])\leq \delta$  holds for any choice of $q_1, q_2 \in [0,1]$. The first and last line hold by (\ref{eq:signatureedge1:eq1}). It remains to show that we can choose $q_1,q_2$ so that $\df(P[t_1,t_2],Q[q_1,q_2]) \leq \delta$ holds. Since $\df(P,X)\leq \delta$, there must be a subsegment $X[x_1,x_2]$ of $X$, such that $\df(P[t_1,t_2],X[x_1,x_2]) \leq \delta$. Recall that $\overline{Q(0)Q(1)}=X$ and by the intermediate value theorem we can define suitable $q_1,q_2$ as follows
\[ q_1 = \max \{q \in [0,1] \mid Q(q)=X(x_1) \} \]
\[ q_2 = \min \{q \in [q_1,1] \mid Q(q)=X(x_2) \} \]
Now, we can apply Lemma~\ref{lemma:montonecurves} and conclude  that $\df(P[t_1,t_2],Q[q_1,q_2]) \leq \delta$. Therefore, also in case (i), we have $\df(P,Q) \leq \delta$.
\end{proof}


Now we are ready to prove Lemma~\ref{lemma:signatureproxy}.

\lemmasignatureproxy*

\begin{proof}
This follows by a modification of the proof of  Lemma~\ref{lemma:simplproxy}. Although the two proofs are very similar, the differences are subtle. Therefore, we give the full proof for the sake of completeness.
Let $q_1,\ldots,q_{\ell}$ be the parameters corresponding to the vertices of $Q'$ in $Q$, i.e., the vertices of $Q'$ are $Q(q_1),\ldots,Q(q_{\ell})$. 
Let $\phi:[0,1] \rightarrow [0,1]^2$ be a $\delta$-traversal between $P$ and $Q'$. Let $0=t_1 \leq \dots \leq t_{\ell}=1$ be a partition of the parameter space of $P$ such that for any $1 \leq i \leq \ell-1$, the edge $\overline{Q(q_i)Q(q_{i+1})}$ is mapped to $P[t_i,t_{i+1}]$ under $\phi$. As such, we have  
\begin{equation}\label{eq:PQ1}
\df (P[t_i,t_{i+1}],\overline{Q(q_i)Q(q_{i+1})}) \leq \delta  
\end{equation}
By the definition of $\delta$-simplifications, we also have that
\begin{equation}\label{eq:PQ2}
 \df (Q[q_i,q_{i+1}],\overline{Q(q_i)Q(q_{i+1})}) \leq \delta' \leq \delta 
\end{equation}
 
Now, if the edge $\overline{Q(q_i)Q(q_{i+1})}$ of $Q'$ is range-preserving, then Lemma \ref{lemma:strongertriangleineq} implies that
\begin{equation}\label{eq:PQ3}
\df(P[t_i,t_{i+1}], Q[q_i,q_{i+1}]) \leq \delta.
\end{equation}
Otherwise, it must be (by the definition of signatures) that either $i=1$ or $i+1=\ell$ or both (the edge is the first or last edge of the signature $Q'$ or $Q'$ consists of just one edge). In any of those cases, Lemma~\ref{lemma:signatureedge1} implies $\df(P[t_i,t_{i+1}], Q[q_i,q_{i+1}]) \leq \delta$.

Finally, we apply Observation \ref{observation:concatenation} on  $P=\bigcirc_{i=1}^{\ell}P[t_i,t_{i+1}]$ and $Q=\bigcirc_{i=1}^{\ell}Q[q_i,q_{i+1}]$, and we obtain  \[\df(P,Q)\leq \max_{i \in [\ell]}\df\left(P[t_i,t_{i+1}], Q[q_i,q_{i+1}] \right) \leq  \delta.\]
\end{proof}

\subsection{Proofs of lemmas for visiting orders}\label{sec:visiting}

In order to prove the existence of $\delta'$-visiting orders for some $\delta' \in O(\delta)$ as claimed in Lemma~\ref{lemma:goodsimplificationranges2}, we introduce the concept of a  visiting sequence. A visiting sequence is not necessarily monotonically increasing, while visiting orders according to Definition~\ref{def:visitingorder} are. Nonetheless, this definition of visiting sequence will turn out to be useful. It is important that a $\delta$-visiting sequence is derived from a monotone traversal. 
We will show (Lemma~\ref{lemma:nonmonotone} and  \ref{lem:visitingorder}) that any non-monotonic visiting sequence can be turned into a monotonic one at the expense of a constant factor in the radius of the visiting sequence.


\begin{Definition}\label{def:visit}
Let $P:[0,1] \rightarrow \RR$ and $Q:[0,1] \rightarrow \RR$ be curves, let $\delta > 0$, and let $\phi: [0,1] \rightarrow [0,1]^2$ be a monotone traversal. We say a vertex $w$ of $Q$ $\delta$-\textbf{visits} a vertex $v$ of $P$ \textbf{under} $\phi$ if the following holds: \begin{compactenum}[(i)] 
\item $|w-v|\leq \delta$ and 
\item at least one of the following holds:
\begin{compactenum}
\item $\phi$ associates $w$ with $v$, or
\item $\phi$ associates $w$ with the interior of an edge of $P$ that is incident to $v$, or
\item $\phi$ associates $v$ with the interior of an edge of $Q$ that is incident to $w$.
\end{compactenum} 
\end{compactenum}
Note that the induced relation on the vertices is symmetric for any fixed $\delta$ and $\phi$. 
\end{Definition}

\begin{Definition}\label{def:visitingseq}
Let $P:[0,1] \rightarrow \RR$ and $Q:[0,1] \rightarrow \RR$ be curves and let $\phi: [0,1] \rightarrow [0,1]^2$ be a monotone traversal. Let $S$ be a subsequence of the vertices of $Q$ of length $\ell$. Let $u_1,\dots,u_{\ell}$ denote the ordered vertices of $S$ and let $v_1,\dots,v_{m}$ denote the ordered vertices of $P$.  A \textbf{$\delta$-visiting sequence} of $S$ on $P$ \textbf{under} $\phi$ is a sequence of indices $i_1,\dots, i_{\ell}$, such that each $u_{j}$ of $S$ $\delta$-visits the vertex $v_{i_{j}}$ of $P$ under $\phi$. 
\end{Definition}

\begin{lemma}\label{lemma:nonmonotone}
Let $P:[0,1]\mapsto \RR$ and $Q:[0,1]\mapsto \RR$ be curves such that $\df(Q,P)\leq \delta$ and let $\phi$ be a monotone traversal realizing this distance. Let $v_i,v_j$ be two vertices of $Q$ with $i < j$ in the ordering along $Q$. Assume $v_i$ $\delta$-visits a vertex $w_a$ of $P$ under $\phi$ and $v_j$ $\delta$-visits  a vertex $w_b$ of $P$ under $\phi$ such that $a > b$ in the ordering along $P$. Then, it must be that $v_i$ $3\delta$-visits $w_b$ under $\phi$ and that $v_j$ $3\delta$-visits $w_a$ under $\phi$.
\end{lemma}

\begin{proof}
As $a > b$, however, in $\phi$ a point on an adjacent edge of $w_a$ is matched earlier than a point on an adjacent edge of $w_b$, we conclude due to the monotonicity of $\phi$ that $\overline{w_b w_a}$ is an edge in $P$.
Let $P(t)$ and $P(t')$ be the points that $v_i$ and $v_j$ are mapped to on $\overline{w_b w_a}$ under $\phi$, respectively.
By the monotonicity of $\phi$ we have $t \leq t'$. See Figure~\ref{fig:caseP1aP2a} for an illustration.

Assume that $w_a < w_b$, as the case $w_a > w_b$ is symmetric. Since $P(t)$ and $P(t')$ are both on the edge $\overline{w_b w_a}$, the fact that $t \leq t'$ implies that $P(t') \leq P(t)$. Using the facts that $|v_i - w_a| \leq \delta$ and $|v_j - w_b| \leq \delta$, we obtain
\[ v_i - \delta \leq w_a \leq P(t') \leq P(t) \leq w_b \leq v_j + \delta.\]
At the same time we have
\[ v_j - \delta \leq P(t') \leq P(t) \leq v_i + \delta.\]

It follows that $|v_i-v_j| \leq 2\delta$.

Thus, the claim that $v_i$ is contained in the $3\delta$-range of $w_b$ is then implied by triangle inequality, as well as the symmetric claim that $v_j$ is contained in  the $3\delta$-range of $w_a$. As $v_i$ and $v_j$ are both matched to the edge $\overline{w_b w_a}$, we also have that $v_i$ and $v_j$ visit the $3\delta$-ranges of $w_b$ and $w_a$, respectively.
\end{proof}

\begin{figure}[t]
    \centering
    \includegraphics{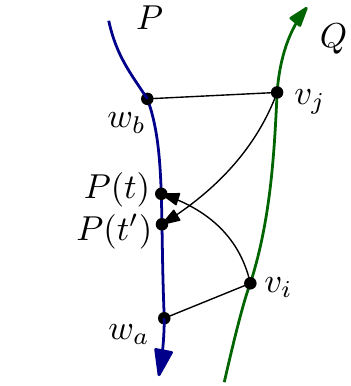}
    \caption{Illustration to the proof of Lemma~\ref{lemma:nonmonotone}. Assuming $w_a<w_b$ as in the proof, $v_i$ visits $w_a$ and $v_j$ visits $w_b$, but $i<j$ and $a>b$, so the visiting sequence is not monotone. }
    \label{fig:caseP1aP2a}
\end{figure}

\begin{lemma}\label{lem:visitingorder}
Let $P:[0,1] \rightarrow \RR$ and $Q:[0,1] \rightarrow \RR$ be curves and let $\phi:[0,1] \rightarrow [0,1]^2$ be a monotone traversal that maps them within distance $\delta$. Let $S$ be a subsequence of the vertices of $Q$. Any $\delta$-visiting sequence of $S$ on $P$ under $\phi$ implies a $3\delta$-visiting order of $S$ on $P$.
\end{lemma}

\begin{proof}
Let $u_1,\dots,u_{\ell}$ denote the vertices of $S$ and let $i_1,\dots,i_{\ell}$ denote the visiting sequence.
We generate a monotonically increasing sequence as follows. For every $u_{j}$, we set $i_{j}$ to the minimum of the suffix sequence $i_{j},\dots,i_{\ell}$. If $i_{j}$ was already a minimum, then nothing changes. Otherwise, let $i_k$ be an index, where this minimum was attained. By  Lemma~\ref{lemma:nonmonotone} the vertex $u_{j}$ is contained in the $3\delta$-range of the vertex $v_{i_k}$. After applying this to all elements of the sequence, starting with $j=1$ and ending with $j=\ell$, the sequence $i_1,\dots,i_{\ell}$ is monotonically increasing.
\end{proof}

The next two lemmas are used in the proof of Lemma~\ref{lemma:goodsimplificationranges2}.

\begin{lemma}
\label{lemma:notvisiting}
Let $P:[0,1]\mapsto \RR$ and $Q:[0,1]\mapsto \RR$ be curves such that $\df(Q,P)\leq \delta$ and let $\phi$ be a monotone traversal realizing this distance. 
If none of the inner vertices of $P$ and $Q$ $\delta$-visit each other under $\phi$, then $P$ and $Q$ are $2\delta$-monotone.
\end{lemma}

\begin{proof}

We prove the lemma by induction. We reconstruct the matching $\phi$ and use \enquote{matched} as shorthand for \enquote{matched under $\phi$}.
Recall that we denote the ordered vertices of $P$ and $Q$ by $p_1, p_2, \dots$ and $q_1, q_2, \dots$, respectively. Note that if either $P$ or $Q$ consist of a single vertex or single segment, then the claim immediately follows from Observation~\ref{lemma:montonecurvesimple}.
Otherwise, either $p_2$ is matched to a point on $\overline{q_1 q_2}$ or $q_2$ is matched to a point on $\overline{p_1 p_2}$ and $p_2, q_2$ are inner vertices. As the lemma statement is symmetric with respect to $P$ and $Q$, we assume without loss of generality that $p_2$ is matched to $\overline{q_1 q_2}$. As $p_2$ and $q_2$ are inner vertices, they cannot $\delta$-visit each other, and thus either $p_2 < q_2 - \delta$ or $p_2 > q_2 + \delta$. By mirroring the curves $P$ and $Q$ at the origin, these two cases are symmetric, and we thus assume $p_2 < q_2 - \delta$ without loss of generality. As $p_2$ is matched to $\overline{q_1 q_2}$, it follows that $q_1 < q_2$. Thus, $\overline{q_1 q_2}$ is increasing and $\overline{p_1 p_2}$ has to be 2\deltamonotone increasing as otherwise the matching would have distance larger than $\delta$.
Now, for the inductive step, assume that $\seqtocurve{ p_1, \dots, p_i }$ and $\seqtocurve{ q_1, \dots, q_j }$ are $2\delta$-monotone increasing curves, $p_i, q_j$ are inner vertices, and $p_i$ is matched to a point on $\overline{q_{j-1}q_j}$ with $p_i < q_j - \delta$. Note that this again implies $q_{j-1} < q_j$.

Let us now prove the inductive step. If $p_{i+1}$ is an inner vertex, then either (i) $p_{i+1}$ is also matched to a point on $\overline{q_{j-1}q_j}$ or (ii) $q_j$ is matched to a point on $\overline{p_i p_{i+1}}$. 

In case (i), $p_{i+1}$ extends a subcurve $\seqtocurve{p_{i'}, \dots, p_i}$ with $i' \geq 1$ that is completely matched to a part of the increasing segment $\overline{q_{j-1}q_j}$. The subcurve $\seqtocurve{p_{i'}, \dots, p_i}$ has to be 2\deltamonotone increasing according to Observation~\ref{lemma:montonecurvesimple}. Either $p_{i'}$ is the start of $P$ (i.e, $i' = 1$) and thus $\seqtocurve{p_1, \dots, p_{i+1}}$ is 2\deltamonotone increasing, or $p_{i'-1}$ has to be matched to a part of $Q$ before $q_{j-1}$ and thus $q_{j-1}$ is an inner vertex.
As $q_{j-1}$ was already matched, it follows that either $p_{i'-1}$ is the start of $P$ (i.e., $i'-1 = 1$) and $p_{i'-1} \leq q_{j-1} + \delta$, or $p_{i'-1}$ is an inner vertex and $p_{i'-1} < q_{j-1} - \delta$ as they do not $\delta$-visit each other.
In both cases $\seqtocurve{p_1, \dots, p_{i'-1}}$ is contained in $[-\infty, q_{j-1}+\delta)$; for the first case this holds as $\seqtocurve{p_1, \dots, p_{i'-1}}$ is 2\deltamonotone increasing by induction. Consequently, the concatenation of $\seqtocurve{p_1, \dots, p_{i'-1}}$ and $\seqtocurve{p_{i'}, \dots, p_{i+1}}$ is also 2\deltamonotone increasing.

Now consider case (ii), i.e., $q_j$ is matched to a point on $\overline{p_i p_{i+1}}$. In this case $\overline{p_i p_{i+1}}$ is increasing as $p_i < q_j$ and $q_j < p_{i+1}$, which is the case because $q_j$ is matched to $\overline{p_i p_{i+1}}$ and $p_i < q_j - \delta$.
Therefore, also in this case it holds that $\seqtocurve{p_1, \dots, p_{i+1}}$ is 2\deltamonotone increasing.
Note that after exchanging $P$ and $Q$, we again fulfill the inductive hypothesis. In particular, since $q_j$ is matched to $\overline{p_i p_{i+1}}$ but $p_{i+1}$ and $q_j$ do not $\delta$-visit each other as both are inner vertices, we must have $q_j < p_{i+1} - \delta$.

Now consider the case that $p_{i+1}$ is not an inner vertex, i.e., it is the last vertex of $P$. In this case, part of $\overline{p_i p_{i+1}}$ has to be matched to $q_j$ as no previous part of $P$ was matched to $q_j$. This implies that $\overline{p_i p_{i+1}}$ again is increasing as $p_i < q_j - \delta$ and $p_{i+1} \geq q_j -\delta$. Hence $\seqtocurve{p_1 \dots p_{i+1}}$ is $2\delta$-monotone increasing. As the remainder of $Q$, starting from $q_j$, has to be matched to part of $\overline{p_i p_{i+1}}$ and therefore this part is $2\delta$-monotone increasing by Observation~\ref{lemma:montonecurvesimple}, and $\seqtocurve{q_1, \dots, q_{j-1}}$ is 2\deltamonotone by induction and also contained in $[-\infty,p_i-\delta)$ as $p_i$ is matched to the increasing  $\overline{q_{j-1} q_j}$, it follows that the whole curve $Q$ is 2\deltamonotone increasing.
\end{proof}

\begin{lemma}
\label{lemma:shortcut}
 Let $P:[0,1]\mapsto \RR$ and $Q:[0,1]\mapsto \RR$ be curves such that $\df(Q,P)\leq \delta$ and let $\phi$ be a monotone traversal realizing this distance. Further assume that for all $t \in [0,1]$ we have $Q(t) \in \overline{Q(0)Q(1)}$.
 If none of the inner vertices of $Q$ $\delta$-visit an inner vertex of $P$ under $\phi$, then the line segment $Q'=\overline{Q(0) Q(1)}$ is a range-preserving $\delta$-simplification of $Q$ with  $\df(Q',P)\leq \delta$. 
\end{lemma}
\begin{proof}
By Lemma~\ref{lemma:notvisiting}, $Q$ and $P$ must be $2\delta$-monotone. Moreover, $Q'$ is range-preserving by assumption. Therefore, $Q'$ is a range-preserving $\delta$-simplification of $Q$.
It remains to show the bound on the Fr\'echet distance of $P$ and $Q'$.
To this end, we want to invoke Observation~\ref{lemma:montonecurvesimple}. Indeed, it must be that 
\[ \forall t\in [0,1]: P(t) \in \bigcup_{s \in [0,1]} B(Q(s),\delta), \] 
since $\df(P,Q) \leq \delta$ and since $Q'$ is range-preserving. Therefore, the conditions of Observation~\ref{lemma:montonecurvesimple} are satisfied and the bound is implied.
\end{proof}

We are now ready to prove Lemma~\ref{lemma:goodsimplificationranges2} from Section~\ref{section:lemmas}.

\lemmagoodsimplificationranges*
\begin{proof}
Let $\phi$ be a monotone traversal that realizes the Fr\'echet distance between $P$ and $Q$. We will construct a $\delta$-\straightening $Q'$  together with a $\Oh(\delta)$-visiting order of $Q'$ on $P$. To this end, consider the subset of vertices of $Q$ that each $\delta$-visit \emph{some} vertex of $P$ under $\phi$ (Definition~\ref{def:visit}). Denote this subset by $S$. Lemma~\ref{lem:visitingorder} implies that there exists a $3\delta$-visiting order of $S$ on $P$.
We denote this visiting order by the function $\kappa: S \rightarrow [m]$ that assigns every vertex of $S$ the index of a vertex of $P$ (where $m$ denotes the number of vertices of $P$). 

It is quite possible that $S$ is not a $\delta$-simplification of $Q$ with the desired properties. In a second phase of the construction we will therefore add more vertices of $Q$ to $S$.
Consider any maximal subcurve $Q[s,s']$ of $Q$, such that none of the inner vertices of $Q[s,s']$ $\delta$-visit a vertex of $P$ under $\phi$. It must be that $Q(s)$ corresponds to some vertex $w$ of $S$ and $Q(s')$ corresponds to some vertex $w'$ of $S$. Moreover, $w'$ comes directly after $w$ along $Q$ among the vertices included in $S$.  Assume that $Q[s,s']$ has at least one inner vertex. We distinguish two cases:
\begin{enumerate}[(C1)]
\item $B(v_{\kappa(w)},3\delta) \cap B(v_{\kappa(w')}, 3\delta) \neq \emptyset $,
\item otherwise
\end{enumerate}

In the first case (C1), we will add all inner vertices $Q[s,s']$ to $S$ and assign them the index $\kappa(w)$ in the constructed visiting order $\kappa$.
In the second case (C2), we will only add a specific subset of vertices, which we define as follows.
Define $\alpha$ and $\beta$ as follows:
\[ \alpha = \max \{ t ~\mid~ t \in [s,s'] \text{ and } Q(t) \in B(v_{\kappa(w)},3\delta)  \} \]
\[ \beta = \min \{ t ~\mid~ t \in [\alpha,s'] \text{ and } Q(t) \in B(v_{\kappa(w')}, 3\delta) \} \]

Since the $3\delta$-ranges of $v_{\kappa(w)}$ and $v_{\kappa(w')}$ are disjoint, $\alpha$ and $\beta$ are well-defined and it follows by definition that $s \leq \alpha \leq \beta \leq s'$. Therefore, the subcurves $Q[s,\alpha]$,  $Q[\alpha,\beta]$, and $Q[\beta,s']$ are well-defined. Now, we proceed as follows, we add the inner vertices of $Q[s,\alpha]$ to $S$ and assign them the index $\kappa(w)$ in the constructed visiting order $\kappa$. Secondly, we add the inner vertices of $Q[\beta,s]$ to $S$  and assign them the index $\kappa(w')$ in the constructed visiting order $\kappa$. 

We apply this to all such maximal subcurves $Q[s,s']$ (note that these are pairwise disjoint), thereby constructing the sequence $S$ along with the visiting order $\kappa$.
Let $u_1,\dots,u_{\ell}$ be the sequence of vertices of the resulting $S$ in their order along $Q$. Denote with $Q'$ the curve that results from linearly interpolating $u_1,\dots,u_{\ell}$. Note that it is different from $Q$ only in the sections where we omitted the vertices of the subcurve $Q[\alpha,\beta]$ in case (C2).
We claim that $Q'$ is an edge-range-preserving $\delta$-simplification of $Q$. To see this, consider a subcurve $Q[s,s']$, assume we are in case (C2). By construction,  the subcurve $Q[\alpha,\beta]$ is range-preserving (for all $x \in [\alpha,\beta]$ we have $Q(x) \in \overline{Q(\alpha) Q(\beta)}$). Let $P[t,t']$ be a subcurve of $P$ mapped to $Q[\alpha,\beta]$ under $\phi$. Now, Lemma~\ref{lemma:shortcut} applied to the subcurves $P[t,t']$ and $Q[\alpha,\beta]$ implies that $\overline{Q(\alpha)Q(\beta)}$ is an  edge-range-preserving $\delta$-simplification of $Q[\alpha,\beta]$ with $\df(\overline{Q(\alpha)Q(\beta)},P[t,t']) \leq \delta$. 
Therefore, by Observation~\ref{observation:concatenation},  when removing all vertices of $Q$ in the parameter range $(\alpha,\beta)$ for each such maximal subcurve $Q[s,s']$, we obtain a $\delta$-\straightening $Q'$ of $Q$ with $\df(Q',P)\leq \delta$. 

Finally, we argue that the constructed visiting order $\kappa(u_1), \dots, \kappa(u_{\ell})$ is an $11\delta$-visiting order of $Q'$ on $P$. Clearly it is monotonically increasing by construction. Also, it is clear that any vertex added in the first phase is contained in the $3\delta$-range of its assigned vertex of $P$. It remains to argue for any vertex added to $S$ in the second phase, that it is contained in the $11\delta$-range of its assigned vertex in $P$.
Consider a subcurve $Q[s,s']$ from above and assume we are in case (C1).  We have that $Q(s) \in B(v_{\kappa(w)},3\delta)$ and $Q(s') \in B(v_{\kappa(w')},3\delta)$. By the case distinction, these two ranges are not disjoint. Therefore, the subcurve starts and ends in the $9\delta$-range of the assigned vertex $v_{\kappa(w)}$.
Moreover, by Lemma~\ref{lemma:notvisiting}, $Q[s,s']$ has to be $2\delta$-monotone. This implies that the entire subcurve lies in the $11\delta$-range of $v_{\kappa(w)}$ and this is also the vertex that we assigned to all of its inner vertices.
A similar argument can be applied in case (C2). By the way we chose $\alpha$, we have that $Q(\alpha)$ is contained in the $3\delta$-range of $v_{\kappa(w)}$, which is also the vertex assigned to the entire subcurve. Since also the subcurve $Q[s,\alpha]$ is $2\delta$-monotone, all remaining vertices in the range $[s,\alpha]$ are contained in the $5\delta$-range of the same vertex. A symmetric argument can be applied to show that all remaining vertices in the range $[\beta, s']$ are contained in the $5\delta$-range of their assigned vertex. 
\end{proof}

Finally, we also prove Lemma~\ref{lemma:signatures3}  from Section~\ref{section:lemmas}.  

\lemmasignatures*
\begin{proof}
By the triangle inequality we have that $\df(P',Q) \leq \df(P',P)+\df(P,Q) \leq 2\delta$. Now Lemma~\ref{lemma:signatures2} applied to $P'$ and the $2\delta$-signature of $Q$ implies that there exists a $2\delta$-visiting order of $Q'$ on $P'$. 

It remains to argue that $\df(P',Q') \le 3\delta$.  Let $\phi:[0,1] \rightarrow [0,1]^2$ be a $\delta$-traversal of $P$ and $Q$. Consider an edge $X$ of $Q'$ and let $Q[\alpha,\beta]$ be the subcurve of $Q$ that corresponds to $X$. 
Let $P[\alpha',\beta']$ be a subcurve of $P$ that is mapped to $Q[\alpha,\beta]$ under $\phi$.
By the triangle inequality  
\[ \df(P[\alpha',\beta'],X) \leq \df(P[\alpha',\beta'],Q[\alpha,\beta]) +  \df(Q[\alpha,\beta]),X)  \leq 3\delta \] 

Assume that $P'$ is range-preserving for now (we will treat the general case below) and let $P'[\alpha'',\beta'']$ be the corresponding subcurve of $P'$ starting at $P(\alpha')$, ending at $P(\beta')$, and with inner vertices being the $\delta$-signature vertices of $P$ in the parametrization interval $[\alpha',\beta']$. Note that $P'[\alpha'',\beta'']$ is well-defined since $P'$ is a range-preserving as assumed above.
By Observation~\ref{observation:shortcut} it follows that 
$\df(P'[\alpha'',\beta''],X) \le 3\delta$.
To show the claim for the case of range-preserving $P'$, we now want to use Observation~\ref{observation:concatenation} to concatenate the corresponding subcurves of $P'$ and $Q'$ and obtain that $\df(P',Q') \le 3\delta$. For this, we can choose the values of $\alpha'$ and $\beta'$ in the above argument such that we obtain a decomposition of $P$ into subcurves. Concretely, let $X_1,\dots,X_s$ be the edges of $Q'$ in their order along $Q'$, with $X_i=\overline{Q(\alpha_i) Q(\beta_i})$.
Then, we can choose the corresponding subcurves of $P$ as $P[\alpha_i',\beta_i']$, with \[ \alpha_{i-1}' \leq \beta_{i-1}'=\alpha_i' \leq \beta_i' \]
for any $1 < i \leq s$, with $\alpha'_1 = 0$ and $\beta'_s = 1$. Thus, we obtain a decomposition of $P$.
Now, if $P'$ is a range-preserving simplification of $P$, then the above construction induces a decomposition of $P'$ into subcurves $P'[\alpha_i'',\beta_i'']$ and we can apply Observation~\ref{observation:concatenation}.

As noted above, $P'$ is not necessarily range-preserving on all edges since it is a signature. In particular, it may not be range-preserving on the first  edge (or the  last edge, or neither). This could lead to $P(\alpha_2')$ (resp. $P(\alpha_s')$ for the last edge) not being included in the image of the signature edge of $P'$ that corresponds to the subcurve of $P$ containing $\alpha_2'$ (resp., $\alpha_s'$). 
Note that if $P(\alpha_2')$ is not contained in the image of the first signature edge, then it must be that $ |P(\alpha_2') - P(0)| \leq \delta $, and in fact, it must be that this holds for the entire subcurve, that is  $|P(t) - P(0)| \leq \delta $ for any $t \in [0,\alpha_2']$.
We claim that in this case we can simply set $\alpha_2''$, and $\beta_1''$ to $0$ (resp., we can set  $\beta_{s-1}''$, and $\alpha_s''$ to $1$). We argue that this way of choosing the decomposition leads to $\df(P'[\alpha_1'',\beta_1''],X_1) \leq 3\delta$ and $\df(P'[\alpha_2'',\beta_2''],X_2) \leq 3\delta$ so that the above arguments can be applied (for the last two edges of $Q'$ a symmetric argument can be applied and we will omit the explicit analysis).

By the triangle inequality, we have that
\[ |P'(0)-Q(\alpha_2)|\leq |P(0)-P(\alpha_2')|+| P(\alpha_2')-Q(\alpha_2)| \leq 2\delta.\] 
Together with \[|P'(0)-Q(\alpha_1)| = |P(0)-Q(0)| \leq \delta \] this implies by Observation~\ref{observation:linesegment} that $\df(P'(0), X_1) \leq 3\delta$ since $X_1$ is a line segment and $X_1=\overline{Q(0)Q(\alpha_2)}$. 
Applying the triangle inequality again, we obtain for any $t \in [0,\alpha_2']$ that 
\[ |P(t) - Q(\alpha_2)| \leq |P(t)-P(0)| + |P(0)-Q(\alpha_2)| \leq 3\delta\] 
By Observation~\ref{observation:concatenation} and since $X_2=\overline{Q(\alpha_2)Q(\beta_2)}$, this implies that \[
\df(P[0,\beta_2'],X_2) \leq  \max \left(~
\df(P[0,\alpha_2'],Q(\alpha_2))~,~  \df(P[\alpha_2',\beta_2'], \overline{Q(\alpha_2)Q(\beta_2)})~ \right) \leq 3\delta\]
By Observation~\ref{observation:shortcut} it follows that 
$\df(P'[0,\beta_2''],X_2) \le 3\delta$.


\end{proof}

\section{Lower Bounds}\label{section:lowerbounds}

In this section we show several conditional lower bounds for $(2-\eps)$ and $(3-\eps)$-approximate nearest neighbor data structures. We use the well-known Orthogonal Vectors problem as the problem that we base our hardness results on.

\begin{Definition}[Orthogonal Vectors (OV)]
Given two sets of vectors $A, B \subseteq \{0,1\}^d$, do there exist two vectors $a \in A, b \in B$ such that $\langle a, b \rangle = 0$?
\end{Definition}

\begin{Definition}[Orthogonal Vectors Hypothesis (OVH)]
For all $\eps > 0$ there exists a $c > 0$ such that there is no algorithm solving OV instances $A, B \subset \{0,1\}^d$ with $|A| = |B|$ and $d = c \log |A|$ in time $\Oh(|A|^{2-\eps})$.
\end{Definition}

The above hypothesis is also sometimes called the Low-Dimensional Orthogonal Vectors Hypothesis and it is implied by the Strong Exponential Time Hypothesis \cite{Williams05}.
We use this version of the Orthogonal Vectors Hypothesis as it allows us to rule out running times using an arbitrarily small $\eps$ while still reducing from an instance where vectors have a logarithmic dimension.
It is well known that \emph{balanced} OV with sets of the same size is equally hard as \emph{unbalanced} OV \cite{AbboudW14, BringmannK18}.

\begin{lemma}[Unbalanced Orthogonal Vectors Hypothesis] \label{lem:unbalancedOVH}
Assume OVH holds true. For every $\alpha \in (0,1)$ and $\eps > 0$ there exists a $c > 0$ such that there is no algorithm solving OV instances $A, B \subset \{0,1\}^d$ with $|B| = |A|^\alpha$ and $d = c \log |A|$ in time $\Oh(|A|^{1+\alpha-\eps})$.
\end{lemma}

\begin{proof}[Proof sketch]
We briefly outline why this hardness holds. To that end, assume that we can solve the unbalanced case in time $\Oh(|A|^{1+\alpha-\eps})$ for some $\eps > 0$. Then we could solve the balanced case by splitting $B$ into $|A|^{1-\alpha}$ parts of size $|A|^\alpha$, solve these instances in time $\Oh(|A|^{1+\alpha-\eps})$, and thus solve the balanced problem in time $\Oh(|A|^{1-\alpha} \cdot |A|^{1+\alpha-\eps}) = \Oh(|A|^{2-\eps})$.
\end{proof}
Leveraging this insight, we later reduce from unbalanced OV instances to show stronger hardness results.
For convenience, we introduce some additional notation. For a vector $a \in \{0,1\}^d$, we use $a[i]$ to refer to its $i$th entry, where the entries are 0-index, i.e., $a = (a[0], \dots, a[d-1])$. Recall that we use the \enquote{$\circ$} operator to concatenate curves and that the curve $P$ where each point is translated by $\tau$ is denoted as $P + \tau$.

Instead of reducing directly from OV, we introduce a novel problem called \ossov and show that it is hard under OV. Subsequently, we reduce from this problem to the ANN problems introduced above.

\subsection{\ossov}
This problem can be thought of as a variant of OV with an additional restriction on one of the input sets. More precisely, for one set we  allow at most $k$ non-zero entries in each vector.
\begin{Definition}[\ossov]
Given a value $k \in \mathbb{N}$ and two sets of vectors $A, B \subseteq \{0,1\}^d$ where each $a \in A$ contains at most $k$ non-zero entries, do there exist two vectors $a \in A, b \in B$ such that $\langle a, b \rangle = 0$?
\end{Definition}

We also refer to \ossov with parameter $k$ as $\ossov(k)$. We now show that this problem is hard under OV, interestingly, this is already the case for $k \in \omega(1)$.

\begin{lemma}\label{lem:ossov}
Assume OVH holds true. For every $\alpha \in (0,1)$, $\eps > 0$ there is a $c > 0$ such that for any $k \in \omega(1) \cap o(\log |A|)$ there is no algorithm solving $\ossov(k)$ instances $A,B \subset \{0,1\}^d$ with $|B| = |A|^\alpha$ and $d = k \cdot |A|^{c/k}$ in time $\Oh(|A|^{1+\alpha-\eps})$. 
\end{lemma}
\begin{proof}
For any $\alpha \le 1, \eps > 0$, let $c > 0$ be the constant from Lemma~\ref{lem:unbalancedOVH}. Thus, unless OVH fails, we cannot solve OV instances $A, B \subset \{0,1\}^d$ with $|B| = |A|^\alpha$ and $d = c \log |A|$ in time $\Oh(|A|^{1+\alpha-\eps})$. For any $k \in \omega(1) \cap o(\log |A|)$, we now reduce to $\ossov(k)$ as follows.
We convert $A$ to a \emph{set of sparse vectors} $A'$ and $B$ to a set $B'$ such that $A', B'$ is an equivalent $\ossov(k)$ instance. To achieve this, we increase the dimensionality of the vectors in the \ossov instance.
Given a vector $a \in A$, partition the dimensions of $a$ into $k$ blocks of size $d/k$.\footnote{If $d$ is not divisible by $k$, increase the dimension until this is the case and fill these dimensions with zeros.}
More precisely, let
\[
	a_i = \left(a\left[i \cdot \frac{d}{k}\right],\, a\left[i \cdot \frac{d}{k} + 1\right],\, \dots,\, a\left[i \cdot \frac{d}{k} + \frac{d}{k}-1\right]\right)
\]
for $i \in \{0,\dots,k-1\}$.
Let $\hat{a}_i \in \left[2^{d/k}\right]$ be defined as the binary vector $a_i$ interpreted as a binary number.
We now construct the corresponding $a' \in A'$ as follows.
We choose the dimension of the vectors in $A', B'$ as $d' = k \cdot 2^{d/k}$ --- note that this equals $k \cdot |A|^{c/k}$ as stated in the lemma.
For each $i \in \{0,\dots,k-1\}$, we set $a'[i \cdot 2^k + \hat{a}_i] = 1$.
All other entries of $a'$ are set to 0.
Thus, each vector $a' \in A'$ contains exactly $k$ $1$-entries.
The vectors $b' \in B'$ we construct as follows.
Given a vector $b \in B$, we also partition its dimensions the same way as we did for $a \in A$ and obtain vectors $b_0, \dots, b_{k-1}$.
For each $i \in \{0,\dots,k-1\}$ and all $\beta \in \{0,1\}^{d/k}$ --- where we again use $\hat{\beta}$ to denote $\beta$ being interpreted as a binary number --- we set $b'[i \cdot 2^k + \hat{\beta}] = 1$ if $\langle b_i, \beta \rangle > 0$, otherwise we set it to zero. This completes the description of the reduction. Note that while we changed the dimension of the vectors, the size of the sets remained the same, that is $|A'| = |A|$ and $|B'| = |B|$.

Note that for any vectors $a \in A$ and $b \in B$ with $\langle a, b \rangle > 0$ there exist parts $a_i, b_i$ and a coordinate $\ell$ such that $a_i[\ell] = b_i[\ell] = 1$, and thus $\langle a_i,b_i \rangle > 0$. Hence, by construction of $b'$, there exists a dimension in $a'$ and $b'$ where both have a 1. On the other hand, if $a'$ and $b'$ contain a 1 in the same dimension, then by construction of $b'$ there have to be two parts $a_i, b_i$ such that $\langle a_i,b_i \rangle > 0$ and thus $\langle a,b \rangle > 0$.

The total running time of this reduction consists of constructing the vectors in $A'$ --- which takes time proportional to the number of entries --- and the inner product computation between vectors of dimensionality $d/k$ for each of the $k \cdot 2^{d/k}$ dimensions of each vector in $B'$:
\[
\Oh\left(|A'| \cdot k \cdot 2^{d/k} + |B'| \cdot k \cdot 2^{d/k} \cdot \frac{d}{k}\right) = \Oh\left(|A| \cdot 2^{c \log |A|/k} \cdot c \log |A|\right) = \Oh\left( |A|^{1 + c/k} \cdot c \log |A|\right),
\]
which simplifies to $|A|^{1+o(1)}$ as $k \in \omega(1)$ and $\log |A| = \Oh(|A|^{o(1)})$.
Thus, if indeed we can solve $\ossov(k)$ in time $\Oh(|A'|^{1+\alpha-\eps})$ and add the running time of the reduction, then we can solve unbalanced OV in time
\[
    \Oh(|A'|^{1+\alpha-\eps}) + |A|^{1+o(1)} = \Oh(|A|^{1+\alpha-\eps}),
\]
which would refute OVH.
\end{proof}

Using this insight, we now proceed to proving hardness results for different approximation ratios for ANN under the continuous Fréchet distance.

\subsection{\boldmath Hardness of $(2-\eps)$-Approximation in 1D}

In this section we present our first hardness result. We note that the gadgets that we use to encode our vectors are inspired by \cite{DP20}.


\begin{theorem}\label{thm:2minusepshard1d}
Assume OVH holds true. For any $\eps,\eps'>0$ there is a $c > 0$, such that there is no $(2-\eps)$-ANN for the continuous Fréchet distance supporting query curves of any complexity $k \in \omega(1) \cap o(\log n)$ and storing $n$ one-dimensional curves of complexity $m = k \cdot n^{c/k}$ with preprocessing time $\poly(n)$ and query time $\Oh(n^{1-\epsilon'})$.
\end{theorem}


\begin{proof}
	We show the hardness by a reduction from $\ossov(k)$. To that end, let $A, B \subset \{0,1\}^d$ be a $\ossov(k)$ instance with $|B| = |A|^\alpha$ for a constant $\alpha \leq 1$ that we specify later, $k \in \omega(1) \cap o(\log |A|)$, and $d = k \cdot |A|^{c/k}$ with a constant $c > 0$ that we later choose sufficiently large.
Recall that, by Lemma \ref{lem:ossov}, there exists a $c > 0$ such that $\ossov(k)$ is OV-hard in this regime.
The goal is to use the $k$-sparsity of the vectors in $A$ to obtain short query curves of length $\Oh(k)$.

Let us first give the reduction. To that end, we define the following subcurves:
\[
	0_A \coloneqq \seqtocurve{0,6}, \quad
	1_A \coloneqq \seqtocurve{0,6,2,6}, \quad
	0_B \coloneqq \seqtocurve{0,5,3,6}, \quad
	1_B \coloneqq \seqtocurve{0,6}
\]
Now, given a $\ossov(k)$ instance $A,B$, we create the input set $\mathcal{P}$ and the query set $\mathcal{Q}$ of a $(2-\eps)$-ANN instance with distance threshold $\delta = 1$ as follows. For each vector $a \in A$, we add the curve $Q_a$ to $\mathcal{Q}$ which is defined as
\[
	Q_a \coloneqq \mathop{\bigcirc}_{i=0}^{d-1} \; V_a^i \quad \text{ with } \quad V_a^i \coloneqq a[i]_A + 6i,
\]
where $a[i]_A$ is either $0_A$ or $1_A$, depending on the value of $a[i]$, and the \enquote{$+6i$} is a translation of each point of the curve by $6i$. For each vector $b \in B$, we add the curve $P_b$ to $\mathcal{P}$ which is defined as
\[
	P_b \coloneqq \mathop{\bigcirc}_{i=0}^{d-1} \; V_b^i \quad \text{ with } \quad V_b^i \coloneqq b[i]_B + 6i,
\]
where $b[i]_B$ is either $0_B$ or $1_B$, depending on the value of $b[i]$.
It is crucial that we make the resulting curves non-degenerate by removing all degenerate vertices. In particular, all connecting vertices between gadget curves will be removed and any sequence of consecutive gadgets $0_A$ will be turned into a single line segment. Thus, the curves in $\mathcal{Q}$ will have complexity $\Oh(k)$. See Figure \ref{fig:2minuseps_lb_1d} for an example of the construction.
\begin{figure}
\centering
\includegraphics[scale=1.1]{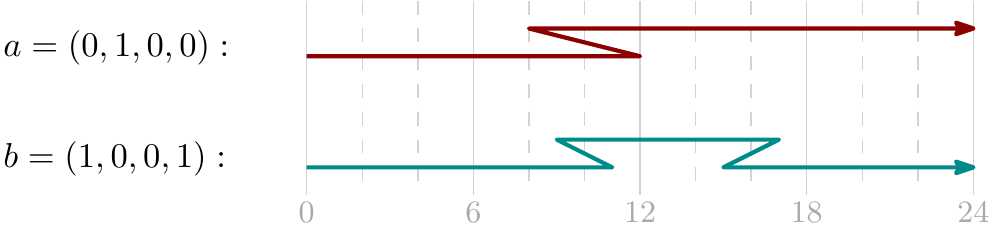}
\caption{Visualization of the $2-\eps$ lower bound in 1D.}
\label{fig:2minuseps_lb_1d}
\end{figure}

We now show correctness of the reduction. Let $P_b \in \mathcal{P}$ and $Q_a \in \mathcal{Q}$ be any curves in these sets. Note that if $\df(P_b, Q_a) < 2$, then if the traversal is a distance $2$ into the gadget $V_a^i$, then we also have to be in the gadget $V_b^i$, as there is no other gadget in distance less than $2$. The same statement holds for $V_a^i$ and $V_b^i$ exchanged. Thus, we traverse the gadgets synchronously. Now consider the case $\langle a,b \rangle = 0$. As $\df(0_A, 0_B) = \df(0_A, 1_B) = \df(1_A, 0_B) = 1$, also $\df(P_b, Q_a) = 1$, as there is no $i \in \{0,\dots,d-1\}$ for which the gadget $V_a^i$ is of type $1_A$ and $V_b^i$ is of type $1_B$. Conversely, consider the case $\langle a,b \rangle = 1$. Then there exist an $i \in \{0,\dots,d-1\}$ such that $V_a^i$ is of type $1_A$ and $V_b^i$ is of type $1_B$. As we traverse the gadgets synchronously and as $\df(V_b^i, V_a^i) = 2$, we have $\df(P_b, Q_a) = 2$. Thus, if we have a $(2-\eps)$-ANN, then we can use it to check if there exist orthogonal vectors $a \in A$ and $b \in B$ by the above reduction.

It remains to show that this reduction implies the claimed lower bound. The time to compute the reduction is linear in the output size and thus negligible.
Recall that $\mathcal{P}$ is the input set, i.e., it is the set that we preprocess, and we run a query for each curve in $\mathcal{Q}$.
Note that by the construction of the above reduction we have $|\mathcal{P}| = |A|^\alpha$, and $|\mathcal{Q}| = |A|$. 
Towards a contradiction, assume that we can solve $(2-\eps)$-ANN with preprocessing time $\Oh(|\mathcal{P}|^{\alpha'})$ for some $\alpha' > 0$ and query time $\Oh(|\mathcal{P}|^{1-\eps'})$ for some $\eps' > 0$.
Choosing $\alpha = 1/\alpha'$, we obtain preprocessing time $\Oh(|\mathcal{P}|^{\alpha'}) = \Oh(|A|^{\alpha \alpha'}) = \Oh(|A|)$ and total query time
\[
	\Oh(|\mathcal{Q}| \cdot |\mathcal{P}|^{1-\eps'}) = \Oh(|A| \cdot |A|^{\alpha (1-\eps')}) = \Oh(|A|^{1 + \alpha - \eps' \alpha}).
\]
Thus, we could solve $\ossov(k)$ in time $\Oh(|A|^{1 + \alpha - \eps' \alpha})$. However, by Lemma \ref{lem:ossov}, there exists a $c > 0$ such that this contradicts OVH.
\end{proof}

\subsection{\boldmath Hardness of $(3-\eps)$-Approximation in 1D}

We now show the first of two hardness results that rule out certain preprocessing and query times for $(3-\eps)$-approximations. Note that ruling out higher approximation ratios is not possible using gadgets that encode the single coordinates, as the distance between the gadgets that encode 1-entries can be at most 3 times the threshold distance due to the triangle inequality between the other gadgets, for details see \cite{BuchinOS19}.
For one-dimensional curves we obtain the following lower bound.
We note that the gadgets that we use to encode our vectors are inspired by \cite{BuchinOS19}.
\begin{theorem}\label{thm:3minuseps_lb_1d}
Assume OVH holds true. For any $\eps,\eps'>0$ there is a $c > 0$, such that there is no $(3-\eps)$-ANN for the continuous Fréchet distance storing $n$ one-dimensional curves of complexity $m$ and supporting query curves of complexity $k$ with $m = k = c \log n$ such that we have preprocessing time $\poly(n)$ and query time $\Oh(n^{1-\epsilon'})$.
\end{theorem}

\begin{proof}
We show the hardness by a reduction from OV. To that end, let $A, B \subset \{0,1\}^d$ be an OV instance with $|B| = |A|^\alpha$ for a constant $\alpha \leq 1$ that we specify later and $d = c \log |A|$ for a constant $c > 0$ that we later choose sufficiently large. 
Recall that, by Lemma \ref{lem:unbalancedOVH}, there exists a $c > 0$ such that this problem is OV-hard.
We now create the input set $\mathcal{P}$ and query set $\mathcal{Q}$ of a $(3-\eps)$-ANN instance with distance threshold $\delta = 1$ as follows.
For convenience, we define the curves
\[
	1_A \coloneqq \seqtocurve{0,6,0},\quad 0_B \coloneqq \seqtocurve{0,7,0},\quad 0_A \coloneqq \seqtocurve{0,8,0},\quad 1_B \coloneqq \seqtocurve{0,9,0}.
\]
First, for each vector $a \in A$ we create a new curve $Q_a \in \mathcal{Q}$ defined as
\[
	Q_a \coloneqq \mathop{\bigcirc}_{i=0}^{d-1} \; V_a^i \quad \text{ with } \quad V_a^i \coloneqq a[i]_A,
\]
where $a[i]_A$ is either $0_A$ or $1_A$, depending on the value of $a[i]$.
Second, for each vector $b \in B$ we create a new curve $P_b \in \mathcal{P}$ defined as
\[
	P_b \coloneqq \mathop{\bigcirc}_{i=0}^{d-1} \; V_b^i \quad \text{ with } \quad V_b^i \coloneqq b[i]_B,
\]
where $b[i]_B$ is either $0_B$ or $1_B$, depending on the value of $b[i]$.
See Figure \ref{fig:3minuseps_lb_1d} for examples of these curves.
\begin{figure}
\centering
\includegraphics[scale=1.1]{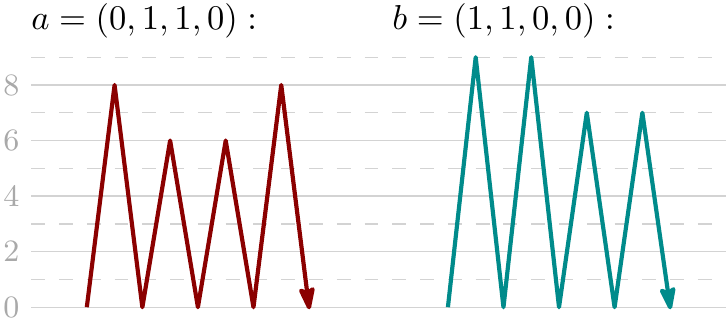}
\caption{Visualization of the $3-\eps$ lower bound in 1D.}
\label{fig:3minuseps_lb_1d}
\end{figure}

We now prove the correctness of the reduction. Consider any $Q_a \in \mathcal{Q}$ and $P_b \in \mathcal{P}$. We first show that if $\df(Q_a, P_b) < 3$, then any traversal realizing this distance has to visit vertices of both curves synchronously. More precisely, a traversal can be in the gadgets $V_a^i$ and $V_b^j$ with $i \neq j$ only if the positions on both curves are strictly less than $6$ in image space. Towards a contradiction, consider the first point in the traversal where this occurs and without loss of generality let the traversal be at position $6$ in $Q_a$. As the traversal on $Q_a$ visited $0$ before, the traversal on $P_b$ has to be below $3$ and thus the positions on $Q_a$ and $P_b$ are within distance more than $3$, which is a contradiction. Thus, when traversing gadgets $V_a^i$ and $V_b^j$ above $6$, then $i = j$.

We now proceed with showing that for all $a \in A$ and $b \in B$ it holds that $\df(Q_a,P_b) \leq 1$ if and only if $\langle a, b \rangle = 0$, and $\df(Q_a,P_b) \geq 3$ otherwise. Assume that $\langle a,b \rangle = 0$, then, by traversing all $V_a^i, V_b^i$ for $i \in \{0,\dots,d-1\}$ synchronously, they can always stay within distance at most 1, as $\df(0_A, 0_B) = \df(0_A, 1_B) = \df(1_A, 0_B) = 1$. However, if $\langle a,b \rangle > 0$, then there exists an index $i \in \{0,\dots,d-1\}$ such that $a[i] = b[i] = 1$. If $\df(Q_a,P_b) < 3$, then we have to traverse these $V_a^i$ and $V_b^i$ synchronously but as $\df(1_A, 1_B) = 3$, there is a point in the traversal where the curves have distance at least $3$ and thus $\df(Q_a,P_b) \geq 3$. It follows that, if we have a $(3-\eps)$-ANN, then it would find if there exists orthogonal vectors in $A$ and $B$ by querying each $Q \in \mathcal{Q}$.

Let us now show that this implies the desired lower bounds. The time to compute the reduction is linear in the output size and thus negligible.
Note that by construction we have $m = k = \Oh(c \log |A|) \le c' \log |A|$ for some constant $c' > 0$. By adding dummy vertices, say many points close to the starting point, we can ensure $m = k = c' \log |A|$ (we could also achieve any intended value $m \ge k$, but this is not necessary for the theorem statement). Moreover, $|\mathcal{P}| = |A|^\alpha$ and $|\mathcal{Q}| = |A|$.
Towards a contradiction, assume that we can solve $(3-\eps)$-ANN with preprocessing time $\Oh(|\mathcal{P}|^{\alpha'})$ for some $\alpha' > 0$ and query time $\Oh(|\mathcal{P}|^{1-\eps'})$ for some $\eps' > 0$.
Choosing $\alpha = 1/\alpha'$, we obtain preprocessing time $\Oh(|\mathcal{P}|^{\alpha'}) = \Oh(|A|^{\alpha \alpha'}) = \Oh(|A|)$ and total query time
\[
	\Oh(|\mathcal{Q}| \cdot |\mathcal{P}|^{1-\eps'}) = \Oh(|A| \cdot |A|^{\alpha (1-\eps')}) = \Oh(|A|^{1 + \alpha - \eps' \alpha}).
\]
Thus, we could solve unbalanced OV in time $\Oh(|A|^{1 + \alpha - \eps' \alpha})$. However, by Lemma \ref{lem:unbalancedOVH}, there exists a $c > 0$ such that this contradicts OVH.
\end{proof}

\subsection{\boldmath Hardness of $(3-\eps)$-Approximation in 2D}
While until here we only considered algorithmic and hardness results for one-dimensional curves, we now show a hardness result for \emph{two-dimensional} curves. This is the only technical section in this paper where we consider two-dimensional curves. Note that in Section \ref{section:prelims} we defined most of our notation for curves in $\RR^d$ and thus the notation of the previous hardness results carries over. For two-dimensional curves we obtain the following lower bound.
\begin{theorem}\label{thm:3minuseps_lb_2d}
Assume OVH holds true. For any $\eps,\eps'>0$ there is a $c > 0$, such that there is no $(3-\eps)$-ANN for the continuous Fréchet distance supporting query curves of any complexity $k \in \omega(1) \cap o(\log n)$ and storing $n$ two-dimensional curves of complexity $m = k \cdot n^{c/k}$ with preprocessing time $\poly(n)$ and query time $\Oh(n^{1-\epsilon'})$.
\end{theorem}


\begin{proof}
	This proof is very similar to the proof of Theorem \ref{thm:2minusepshard1d}. The significant difference is the gadgets that we construct. To this end, consider a $\ossov(k)$ instance $A,B$, where we again use the $k$-sparsity of the vectors in $A$ to obtain short query curves of length $\Oh(k)$.
	We define the generic subcurve
\[
	V(y) \coloneqq \seqtocurve{(0,0), (3,0), (3,y), (6,y), (6,0)}
\]
to then define the usual gadgets
\[
	0_A \coloneqq V(0),\quad
	1_A \coloneqq V(2),\quad
	0_B \coloneqq V(1),\quad
	1_B \coloneqq V(-1).
\]
Now, given a $\ossov(k)$ instance $A,B$, we create the input set $\mathcal{P}$ and query set $\mathcal{Q}$ of a $(3-\eps)$-ANN with distance threshold $\delta = 1$ as follows. For each vector $a \in A$, we add the curve $Q_a$ to $\mathcal{Q}$ which is defined as
\[
	Q_a \coloneqq \mathop{\bigcirc}_{i=0}^{d-1} \; V_a^i \quad \text{ with } \quad V_a^i \coloneqq a[i]_A + (6i,0),
\]
where $a[i]_A$ is either $0_A$ or $1_A$, depending on the value of $a[i]$, and the \enquote{$+(6i,0)$} is a translation of each point of the curve by $(6i,0)$. For each vector $b \in B$, we add the curve $P_b$ to $\mathcal{P}$ which is defined as
\[
	P_b \coloneqq \mathop{\bigcirc}_{i=0}^{d-1} \; V_b^i \quad \text{ with } \quad V_b^i \coloneqq b[i]_B + (6i,0),
\]
where $b[i]_B$ is either $0_B$ or $1_B$, depending on the value of $b[i]$.
It is crucial that we make the resulting curves non-degenerate by removing all degenerate vertices. 
In particular, any sequence of consecutive gadgets $0_A$ will be turned into a single line segment. Thus, the curves in  $\mathcal{Q}$ will have complexity $\Oh(k)$. See Figure \ref{fig:3minuseps_lb_2d} for an example of the construction.
\begin{figure}
	\centering
	\includegraphics[scale=1.1]{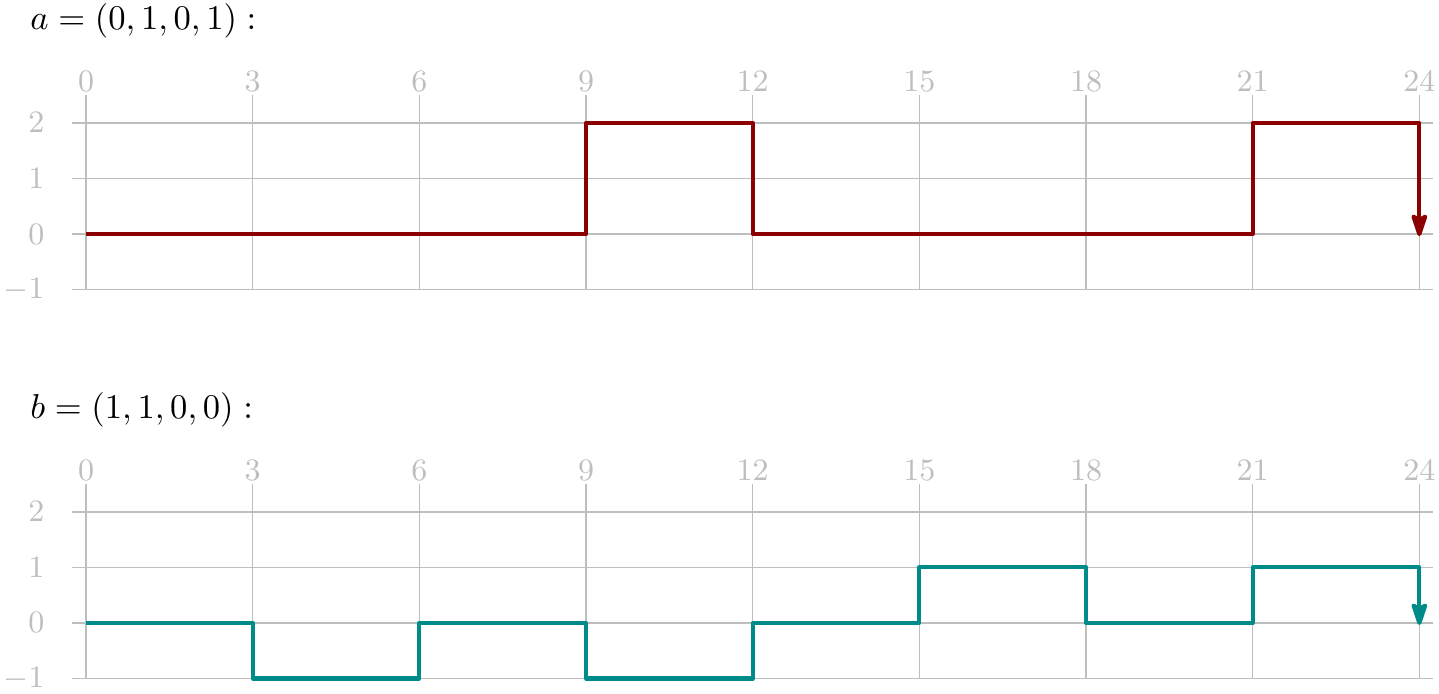}
	\caption{Visualization of the $3-\eps$ lower bound in 2D.}
	\label{fig:3minuseps_lb_2d}
\end{figure}

We now prove correctness of the reduction. Consider the case of two orthogonal vectors $a \in A$ and $b \in B$ such that there is an $i \in \{0,\dots,d-1\}$ with $a[i] = b[i] = 1$. Note that for $\df(Q_a, P_b) < 3$, there has to be a point in the traversal where we are in some point $(x_a, 2)$ in $Q_a$ and in some point $(x_b, -1)$ in $P_b$ as otherwise the distance of the $x$-coordinate would be at least 3. However, the $y$-distance of these points is 3 and thus $\df(Q_a, P_b) \geq 3$. On the other hand, if $a \in A$ and $b \in B$ are orthogonal, then we can traverse the two curves with the same speed in $x$-direction --- i.e., staying at the same $x$-coordinate at every point in time --- and obtain a Fréchet distance at most 1 as $\df(0_A, 0_B) = \df(0_A, 1_B) = \df(1_A, 0_B) = 1$, where the described traversal realizes these distances.

The remainder of the proof, i.e., the derivation of the claimed lower bound, is the same as in the proof of Theorem \ref{thm:2minusepshard1d} and we thus omit it for brevity.
\end{proof}

\section{Conclusions and Open Problems}\label{sec:discussion}
In this work we largely resolve the $\alpha$-ANN problem under the continuous Fréchet distance for one-dimensional curves from a fine-grained perspective for $1 < \alpha < 3$.
We show that, in general, most of the running times presented in this work cannot be improved significantly, however, other tradeoffs between preprocessing time and query time are still possible, and other parameter regimes might be shown hard or more tractable, e.g., for $k \in \Oh(1)$. Indeed, there is a line of work on related data structure problems using the continuous Fr\'echet distance for the specific value of $k=2$, which corresponds to queries with line segments, see~\cite{BergCG13, DBLP:journals/corr/abs-2102-05844}.
It also remains a fundamental problem to show fine-grained lower bounds for approximation factor larger than 3 for a metric problem, which seems to require fundamentally different techniques, cf.~\cite{Rubinstein18}.

As for the continuous Fr\'echet distance, our new upper and lower bounds show that the case of one-dimensional curves provides a kaleidoscopic view into the computational complexity and the underlying challenges posed by the general problem for polygonal curves in $\RR^d$. 
The obvious way forward in this line of research is to show upper and lower bounds for dimension~2 and higher. Some of our ideas might translate directly, such as the idea to generate candidate curves at query time in order to achieve a tradeoff between preprocessing and query time. 
While our lower bounds also hold in higher dimension, it is conceivable that higher lower bounds can be shown already in the plane.
In fact, we already initiate this line of work by showing an equally high lower bound for $(3-\eps)$-ANN in the plane as we have for $(2-\eps)$-ANN for one-dimensional curves.
This lower bound already hints at techniques that can potentially achieve a matching upper bound. We leave this as an open problem.
Our notions of straightenings and signatures, which capture the approximate shape of one-dimensional curves in a best-possible way, currently do not exist in dimension 2 or higher. Extending these notions to the plane by itself would be very interesting.

\typeout{}
\bibliographystyle{alpha}
\bibliography{main}

\end{document}